\newcommand{\Tr}[1]{\mathrm{Tr} #1 }
\newcommand{\TS}{\mathrm{TS}}

\documentclass[aps,pra,twocolumn,superscriptaddress,nofootinbib]{revtex4-1}
\usepackage[english]{babel}
\usepackage{amsthm}
\usepackage{braket}
\usepackage{dsfont}
\usepackage{graphicx}
\usepackage{epsfig}
\usepackage{amsmath,amssymb,amsfonts}
\usepackage{amsthm}
\usepackage{graphicx}
\newtheorem{theorem}{Theorem}

\newtheorem{proposition}[theorem]{Proposition}

\begin{document}
\keywords{Tree size, complexity, quantum computation}

\title{State complexity and quantum computation}

\author{Yu \surname{Cai}}
\affiliation{Centre for Quantum Technologies, National University of Singapore, 3 Science Drive 2, Singapore 117543, Singapore}
\author{Huy Nguyen \surname{Le}}
\affiliation{Centre for Quantum Technologies, National University of Singapore, 3 Science Drive 2, Singapore 117543, Singapore}
\author{Valerio \surname{Scarani}}
\affiliation{Centre for Quantum Technologies, National University of Singapore, 3 Science Drive 2, Singapore 117543, Singapore}
\affiliation{Department of Physics, National University of Singapore, 2 Science Drive 3, Singapore 117542, Singapore}

\begin{abstract}
The complexity of a quantum state may be closely related to the usefulness of the state for quantum computation. We discuss this link using the tree size of a multiqubit state, a complexity measure that has two noticeable (and, so far, unique) features: it is in principle computable, and non-trivial lower bounds can be obtained, hence identifying truly complex states. In this paper, we first review the definition of tree size, together with known results on the most complex three and four qubits states. Moving to the multiqubit case, we revisit a mathematical theorem for proving a lower bound on tree size that scales superpolynomially in the number of qubits. Next, states with superpolynomial tree size, the Immanant states, the Deutsch-Jozsa states, the Shor's states and the subgroup states, are described. We show that the universal resource state for measurement based quantum computation, the 2D-cluster state, has superpolynomial tree size. Moreover, we show how the complexity of subgroup states and the 2D cluster state can be verified efficiently. The question of how tree size is related to the speed up achieved in quantum computation is also addressed. We show that superpolynomial tree size of the resource state is essential for measurement based quantum computation.  The necessary role of large tree size in the circuit model of quantum computation is still a conjecture; and we prove a weaker version of the conjecture.
\end{abstract}

\maketitle

\section{Introduction}
While we all have a feeling of what ``complex'' means, it is notoriously hard to find quantitative measures. Furthermore, there are various types of complexity. The three main examples in present-day science seem to be computational complexity, process complexity, and state complexity. Computational complexity refers to the amount of resources required to perform a certain computation task, be it in terms of time, memory space or number of queries, contributing to different complexity classes in computer science~\cite{arora09}. Process complexity is often associated with the chaotic (but not random) behavior of the process, interconnectivity of many components in the process, and possibly the phenomenon of emergence~\cite{mitchell}. Finally, state complexity, the focus of this paper, refers to the amount of information that is required to describe, generate or simulate a state of a physical system. 

Why do we study the complexity of quantum states? From a foundational point of view, complexity could be a parameter to test the limits of quantum mechanics. The direct extension of quantum effects (coherent superpositions) to daily objects might result in bizarre paradoxes, as Schr{\"o}dinger famously noticed. It is a current experimental trend to push the tests of quantum mechanics towards the macroscopic domain, see for example Refs.~\cite{Arndt14,Arndt03,Friedman00,Cleland10}. In all these experiments, the superposition indeed involves large number of particles or excitations, nevertheless the states produced are somewhat "simple": some involve the superposition of a single degree of freedom, the center of mass; others target the GHZ state, or the Dicke state with few excitations, as ideal macroscopic states. The macroscopic objects of our daily experience are not only large in size, mass and number of particles, but at the same time also interconnected in a non-trivial manner: a cat, besides being large, is a complex object. Do complex objects still obey quantum physics? If yes, as most physicists would argue, can we create them in a controlled way? These questions loomed behind the discussion on the possibility of large-scale quantum computation. In order to refine this discussion, Aaronson~\cite{aaronson2004multilinear} took a technical step and proposed the concept of tree size (TS), as a measure of complexity for quantum states. This highlights that, besides exploring the limits of quantum mechanics, quantum state complexity is a way of capturing the deep relation between complexity and computation. 

The origin of quantum speed up might be sought in some features of entanglement. Promisingly, early studies showed that states used in various quantum algorithm display multipartite entanglement~\cite{bruss2011}; while states with little entanglement could be efficiently simulated with classical computing~\cite{vidal2003efficient, vandennest2007classical}. Nonetheless, large entanglement is neither necessary~\cite{vandennest2013universal} nor sufficient condition for quantum speed up: to the contrary, having too much entanglement might be detrimental to performing computation~\cite{gross2009most}. Measures of entanglement developed with other operational meanings do not seem to capture the computational power of the state.

Another candidate is the phenomena of interference. Previous works~\cite{braun2006quantitative} propose how to quantify interference with ``ibits'' and investigate how many ibits were ``actually used'' in various quantum algorithms. The different amount of 
actually used ibits seems to explain the different amount of speed up in Shor's and Grover's algorithm. The relation with success probability in algorithm with imperfections was studied in~\cite{braun2008interference}.

In this work, we look into the relation between complexity of quantum states and quantum speed up. It seems intuitive that, \textit{in order to be useful in computational tasks, a state must be complex to describe and yet be simple to prepare}. Indeed, if on the one hand a state is simple to describe, it should be possible to simulate it efficiently with classical computers; on the other hand, the preparation of the state from easily available resources is part of the overall computation process. Here, we focus on the first aspect: how to quantify \textit{the complexity of describing a quantum state}?

Among the different measures of state complexity, quantum Kolmogorov complexity is defined by length of the shortest possible program that would generate the state~\cite{berthiaume01,mora05,mora07,rogers08}. This very common definition suffers from the setback that it is not computable. Moreover, Kolmogorov complexity captures the complexity of \textit{generating} the state. The tree size (TS) complexity that we mentioned before and that we are about to discuss relates more to the \textit{description and simulation complexity} of quantum states. The most common way to represent quantum states is the Dirac notation. Tree size complexity can be understood as the size of the minimal description using this notation. 

This article provides a concise summary of our knowledge of the tree size, as well as some new results on complex states with superpolynomial tree size, verification of complex states, and the connection between tree size complexity and the power of quantum computation. The definition of tree size complexity is given in Sec.~2. Next, we review the works on states of two, three and four qubits, and describe the most complex states according to this measure. Moving to the case of $n$ qubits, we first discuss a few examples of simple states with polynomial tree size. In Sec.~4, a theorem by Raz for showing superpolynomial lower bound on multilinear formula sizes, which in turn lower bounds tree size, is revisited. With this theorem at hand, we show some families of states with superpolynomial tree size: the Immanant states, the Deutsch-Jozsa states and the subgroup states. Based on numerical evidence, we construct an explicit example of a subgroup state with superpolynomial tree size. More importantly, the tree size of the 2D cluster state is shown to be superpolynomial. In Sec.~5, we describe how to verify the superpolynomial tree size of the complex subgroup states and the 2D cluster state with polynomial effort via measuring a witness. The possible relation between state complexity and quantum computation speed up is discussed in Sec.~6. Finally, we offer a list of open problems and technical conjectures.

\section{Tree Size}

\subsection{Definition and basic properties}
Just as bits to classical information, qubits are the basic building blocks of quantum information. Any $n$-qubit pure states can be written in the computational basis with at most $2^n$ coefficients:
\begin{equation}
\ket{\psi} = \sum_{x \in \{ 0,1 \}^n} c_x \ket{x},
\end{equation}
where each $n\times 1$ vector $x$ in $\{ 0,1 \}^n$ is identified as a bit string in $\ket{x}$.
This decomposition on a computational basis is not the most compact way of writing a pure state. In the case of two qubits, the most economic representation is given by the Schmidt decomposition. This decomposition can be iterated to deal with multi-partite states, but already for three qubits a different \textit{ad hoc} decomposition is more compact \cite{Acin00}. The shortest possible representation of a multiqubit state in Dirac notation is given by the \textit{minimal tree size} introduced by Aaronson in Ref.~\cite{aaronson2004multilinear} as a measure of complexity of a pure state: Any multiqubit state written in Dirac's notation can be described by a rooted tree of $\otimes$ and $+$ gates; each leaf vertex is labelled with a single qubit state $\alpha \ket{0}+\beta\ket{1}$, which needs not be normalized. The three-qubit biseparable state $\ket{0}(\ket{00}+\ket{11})$, for instance, is represented by the tree in Fig.~\ref{TB}. The size of a tree is defined as the number of its leaves: thus, the size of the tree of Fig.~\ref{TB} is five. A given state can have many different tree representations (for instance, the biseparable state in Fig.~\ref{TB} can also be written as $\ket{000}+\ket{011}$ whose size would be six). The minimal tree of a state $\ket{\psi}$ is the tree with the smallest size that describes it; and the tree size $\TS(\ket{\psi})$ is the size of the minimal tree.

This measure of complexity is in principle \textit{computable}, though we lack efficient algorithms. Moreover, a relation with multilinear formulas leads to \textit{lower bounds} on the tree size. It is thus possible to show that the tree size of some states is definitely superpolynomial in the number of qubits $n$ \cite{aaronson2004multilinear,us13}. Such states can be considered genuinely complex, in the sense that they cannot have a polynomial (i.e. computationally efficient) representation in Dirac notation --- nor in a matrix-product representation with matrices of constant size \cite{us13}. In contrast, if the TS scales \emph{polynomially} with $n$, the state can be considered simple as it can be described efficiently on a classical computer. Before moving to the $n$-qubit case, let us familiarize ourselves with $\TS$ by looking at the states of a few qubits.
\begin{figure}[t]
	\centering
	\includegraphics[scale=0.25]{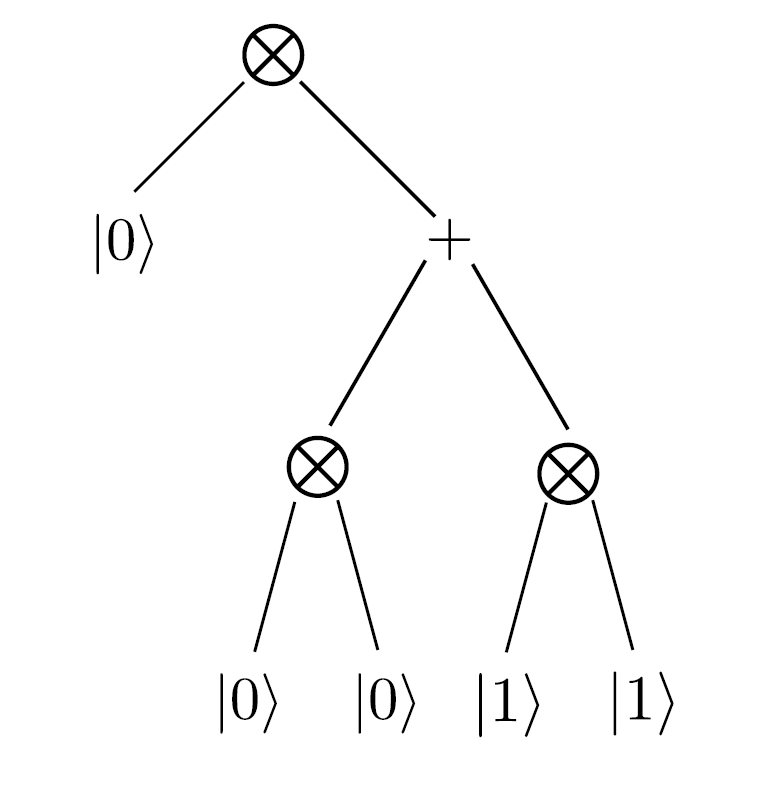}
	\caption{A rooted tree of a three-qubit biseparable state.}\label{TB}
\end{figure}
\subsection{Tree size of few-qubit states}
One observation that is very useful for finding $\TS$ of few-qubit states is the fact that $\TS$ is invariant under invertible local operations (ILOs). Formally \cite{us13},

\begin{proposition}
	\label{thm:ILO}
	If $\ket{\psi} = A_1 \otimes \cdots \otimes A_n \ket{\phi}$, where all the single-qubit operators $A_i$ are invertible, then $\TS(\ket{\psi}) = \TS(\ket{\phi})$.
\end{proposition}

Any two states that can be transformed to each other by ILOs as above are said to be equivalent under stochastic local operation and classical communication (SLOCC). The above proposition implies that all states in a SLOCC equivalent class have the same $\TS$.

\paragraph{Two qubits:} Any two-qubit state can be written in the Schmidt decomposition as \cite{nielsen2000quantum}:
% % % %
\begin{eqnarray}
\ket{\psi}=c\ket{0}\otimes\ket{0}+s\ket{1}\otimes\ket{1},
\end{eqnarray}
% % % %
where $c$ and $s$ are nonnegative real numbers satisfying $c^2+s^2=1$, and $\{\ket{0},\ket{1}\}$ form an orthonormal basis. The state is said to be separable if one of the coefficients $c$ or $s$ vanishes and entangled otherwise.  The Schmidt decomposition has size at most $4$, hence the $\TS$ of any two-qubit state is at most $4$. There are only two different rooted trees of size at most $4$ that can describe a two-qubit state, which are shown in Fig.~\ref{F2}. From this figure we see that a two-qubit state has $\TS=4$ if it is entangled and $\TS=2$ if it is separable. This concludes the case for two qubits. %As an example, the Bell state $(\ket{00}+\ket{11})/\sqrt{2}$ has $\TS=4$ and is the most complex two-qubit state.

\begin{figure}[t]
	\centering
	\includegraphics[scale=0.3]{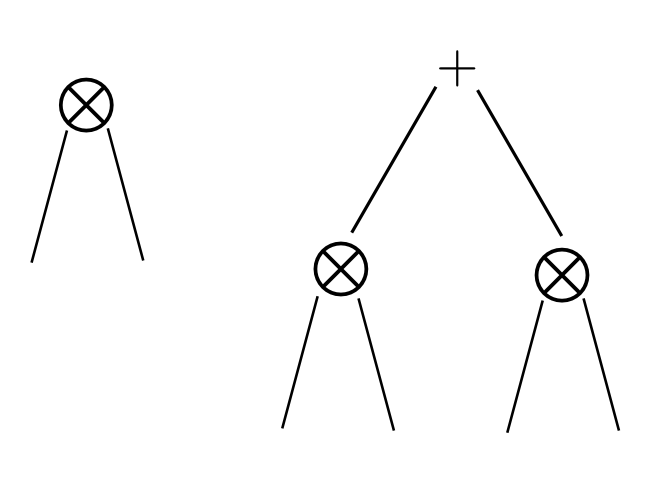}
	\caption{Possible rooted trees with size at most 4 for two-qubit states.}\label{F2}
\end{figure}

\paragraph{Three qubits:} For three qubits, a useful decomposition that has a similar role as the Schmidt decomposition does for two qubits is the canonical form derived by Ac{\'i}n \textit{et al.} \cite{Acin00}: Any three-qubit state can be written as
\begin{eqnarray}
\ket{\psi}=\cos \theta \ket{000}+\sin \theta \ket{1} \left(\cos \omega \ket{0'0''}+\sin \omega \ket{1'1''}\right),
\end{eqnarray}
where the prime and double prime indicate different bases. The $\TS$ is upper bounded by the size of this decomposition and thus is at most $8$. Similarly to the case of two-qubit states, we first find all the possible trees for three qubits with at most eight leaves, and then try to see which one is the minimal tree of a given state.

As stated in Proposition~\ref{thm:ILO}, since all the states in a SLOCC class has the same $\TS$, we need only find $\TS$ of one state in a class to know $\TS$ of all the states in that class. For three qubits, it is known that the pure states can be categorized into six different classes: the product class, three biseparable classes due to permutation, the GHZ class and the W class \cite{Dur00}. Examples of states in these classes are, respectively,
\begin{align}
\ket{\mathrm{P}}=&\ket{000}, \nonumber \\
\ket{\mathrm{B}}=&\frac{1}{\sqrt{2}}\ket{0}\left(\ket{01}+\ket{10}\right),  \nonumber \\
\ket{\mathrm{GHZ}}=&\frac{1}{\sqrt{2}}\ket{000}+\ket{111},  \nonumber \\
\ket{\mathrm{W}}=&\frac{1}{\sqrt{3}}\left(\ket{001}+\ket{010}+\ket{100}\right).
\end{align}
So, a state $\ket{\psi}$ is said to be in a particular SLOCC class, say the W class, if there exist ILOs $A_1,A_2,A_3$ such that $\ket{\psi} = A_1 \otimes A_2 \otimes A_3 \ket{\mathrm{W}}$.

An exhaustive search \cite{us14} shows that $\TS(\ket{\mathrm{W}})=8$, $\TS(\ket{\mathrm{GHZ}})=6$, $\TS(\ket{\mathrm{B}})=5$, and $\TS(\ket{\mathrm{P}})=3$. So the $\TS$ of a three-qubit state can  adopt only one of these four different values depending on what SLOCC class the state belongs to. 
\begin{figure}[tb]
	\centering
	\includegraphics[scale=0.25]{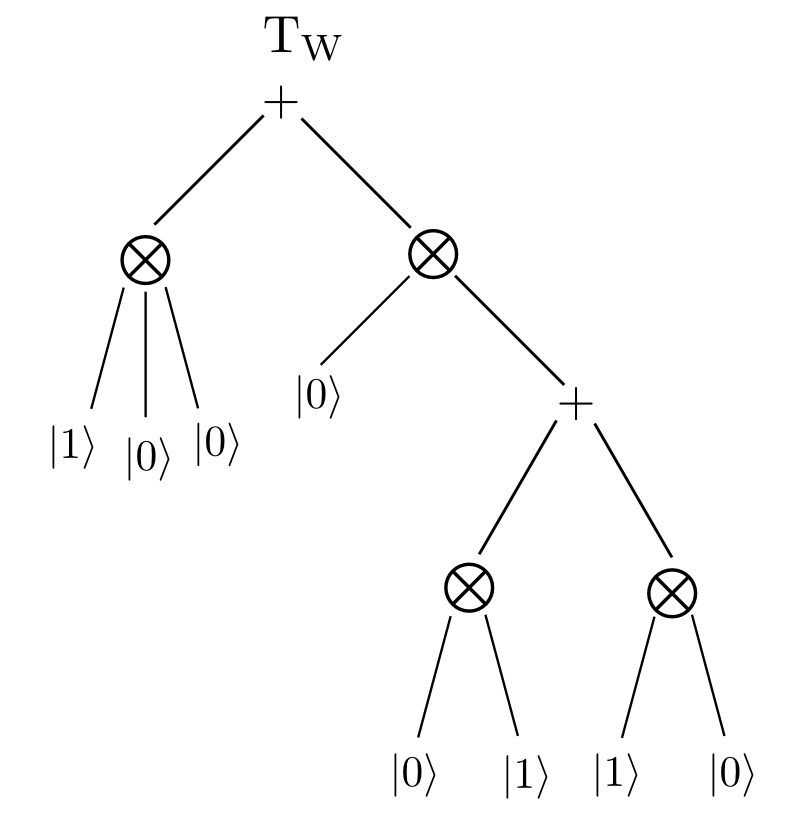}
	\caption{Minimal tree of the (unormalized) $\ket{\mathrm{W}}$ state}\label{TW}
\end{figure}

The most complex three-qubit states are obviously the ones in the W class, whose minimal tree is drawn in Fig.~\ref{TW}. Interestingly, this complexity class is unstable in the sense that an arbitrarily small deviation could bring the W state to a state in the GHZ class~\cite{us14,eisert2006}. For studying how $\TS$ changes in the presence of fluctuation, we define a smoothed version of the tree size over a small neighbourhood of the desired state. For a positive constant $\epsilon<1$, the $\epsilon$-approximate tree size of $\ket{\psi}$ is the minimal tree size over all pure states $\ket{\phi}$ such that $|\braket{\phi|\psi}|^2\geq 1-\epsilon$, that is,
\begin{eqnarray}
\label{eqn:epsilon}
\TS_\epsilon(\ket{\psi}) = \min_{|\braket{\phi|\psi}|^2 \geq 1-\epsilon} \TS(\ket{\phi}).
\end{eqnarray}
Since in practice we cannot know a state with arbitrary precision, the $\TS_\epsilon$ is a more physical measure. The instability of the $\ket{\mathrm{W}}$ state could be now phrased precisely as follows: For $\epsilon$ arbitrarily small, there exists a state $\ket{\phi_{\mathrm{GHZ}}}$ in the GHZ class such that $|\braket{\phi_{\mathrm{GHZ}}|\mathrm{W}}|^2\geq 1-\epsilon.$ Therefore, $\TS_\epsilon(\ket{\mathrm{W}}) = 6$.

\paragraph{Four qubits:} As in the case of three qubits, SLOCC equivalent classes can be used to find the tree size of four-qubit states. In Ref.~\cite{us14} it is shown that the maximal $\TS$ of four-qubit state is 16 and a set of criteria for identifying whether a given state has this maximal $\TS$ is given. The class of most complex four-qubit states were found to belong to a SLOCC class not described in previous inductive classifications \cite{Lamata07}. These states have the following minimal decomposition: 
\begin{eqnarray}\label{eqn:4qubitdecomp}
\ket{\psi} = \ket{\phi_{12}}\ket{\varphi_{34}} + \ket{\phi'_{13}}\ket{\varphi'_{24}},
\end{eqnarray}
where $\ket{\phi}$ and $\ket{\varphi}$ are two-qubit entangled states. Noting the switching of order of qubits in the second branch, this form seems to preclude a recursive construction of the most economic description in terms of tree size. Forms that look recursive do require 18 leaves for some states. 

An example of the most complex four-qubit states with $\TS=16$ is  
\begin{align}
\ket{\psi^{(4)}}=\sqrt{\frac{1}{3}}\bigg[&\frac{1}{2}\left(\ket{0110}+\ket{0101}+\ket{1001}+\ket{1010}\right)\nonumber \\
&-\ket{0011}-\ket{1100}\bigg].
\end{align}
In this computational basis expansion its size is $24$, but it can be shown that its minimal decomposition has indeed the form given in Eq.~\eqref{eqn:4qubitdecomp} with size $16$. This state has already been created in experiments with four-photon down conversion~\cite{Bourennane04,Eibl03}.

When fluctuation is taken into account, the maximal $\epsilon$-tree size of four-qubit states reduces to 14 for $0<\epsilon<1/12$: Any state in the most complex class can be $\epsilon$-approximated by a state with has the decomposition $\ket{0}\ket{\mathrm{GHZ}} + \ket{1} \ket{\mathrm{GHZ}'}$ where $\mathrm{GHZ}$ and $\mathrm{GHZ}'$ are two states in the $\mathrm{GHZ}$ class of three-qubit states.

\paragraph{Mixed states:} The concept of tree size may be extended to a mixed state $\rho = \sum_i p_i \ket{\psi_i}\bra{\psi_i}$ as follows:
\begin{eqnarray}
\TS(\rho) = \min \{ \max_i \TS(\psi_i) \},
\end{eqnarray}
where the minimization is done over all the possible pure state decomposition of the mixed state $\rho$. This is the same approach for extending an entanglement measure to mixed states discussed in Ref.~\cite{Terhal00}. The intuition behind this definition is that the tree size of a mixed state should be at least as complex as the most complex pure state in its decomposition. 

According to  the classification of three-qubit mixed sates introduced by Ac{\'i}n \textit{et al.} \cite{Acin01}, there are four different SLOCC classes: Class $\mathrm{S}$, the set of states that can be written as combination of pure separable states; Class $\mathrm{B}$, for states that can be written as combination of separable and biseparable states; Class $\mathrm{W}$, for states that can be written as combination of separable, biseparable and $\mathrm{W}$ states; and class GHZ, for states that can be written as combination of all possible three qubit states. Clearly, $\mathrm{S} \subset \mathrm{B} \subset \mathrm{W} \subset \mathrm{GHZ}$. From the definition of $\TS$ for mixed states one sees that states belongs to the class $\mathrm{S}$ have tree size 3, $\mathrm{B} \setminus \mathrm{S}$ tree size 5. For states in $\mathrm{GHZ} \setminus \mathrm{B}$, the tree size is 8 if a $\mathrm{W}$ state is required in the decomposition, otherwise it is 6. As an example we look at a family of one-parameter mixed state, the so called generalized Werner state \cite{Eltschka12},
\begin{eqnarray}
\rho(p) = p\ket{\mathrm{GHZ}} \bra{\mathrm{GHZ}} + (1-p) \frac{\mathds{1}}{8}.
\end{eqnarray}
By looking at what SLOCC class $\rho(p)$ belongs to for different values of $p$, it is shown that $\TS = 3$ ($\rho \in \mathrm{S}$) when $p\leq 1/5$, $\TS = 5$ ($\rho \in \mathrm{S} \setminus \mathrm{B}$) when $1/5 < p \leq 3/7$.  With an obvious decomposition into $\mathrm{GHZ}$ and product states, $\TS=6$ for $3/7 < p \leq 1$, even though $\rho \in \mathrm{W \setminus \mathrm{B}}$ when $3/7 < p \leq p_w$ and $\rho \in \mathrm{GHZ} \setminus \mathrm{W}$ when $p > p_w$, where $p_w \approx 0.695\;5427$ \cite{us14}.

\section{Simple states}
We now move to the case of $n$-qubit states and consider how $\TS$ scales as $n\rightarrow \infty$. A state (more precisely, a family of states indexed by $n$) is simple if its $\TS$ scales polynomially with the number of qubits $n$. For showing that a state is simple, it suffices to find an explicit decomposition with polynomial size. Some examples of simple states are given in Table \ref{tab:simple}.

\begin{table}[hb]
	\caption{Summary of $n$-qubit simple states}
	\label{tab:simple}
	\begin{center}
		\begin{tabular}{c|c}
			%a & b & c\\
			%c & d & e
			Product state  &  $n$ \\
			GHZ$_n$ states		&  $2n$ \\
			%		W$_n$ states		&  $O(n^2)$ \\
			Dicke states $D_{n,k}$ & $O(n^2)$ \\
			MPS with rank $\chi$ &  $O(n^{\log 2\chi})$
		\end{tabular}
	\end{center}
	
\end{table}

The product state $\ket{0}^{\otimes n}$ is the simplest state in terms of tree size. It is usually regarded as the input for the circuit model of quantum computation. Obviously $\TS(\ket{0}^{\otimes n}) = n$, which is the minimal $\TS$ for $n$-qubit states. 

The $n$-qubit GHZ state, $(\ket{0}^{\otimes{n}}+\ket{1}^{\otimes{n}})/\sqrt{2}$, which saturates most of the macroscopicity measures \cite{Frowis12}, has $\TS(\ket{\mathrm{GHZ}_n}) = 2n$, which is linear in the number of qubits. This is a clear evidence that complexity is a different notion from macroscopicity. A maximally macroscopic state can yet be very simple.

The Dicke states $\ket{D_{n,k}}$ represents the equal superposition of $n$-qubit string with $k$ excitations; formally it is the (unnormailized) uniform superposition of all the $n$-bit strings with Hamming weight $k$: 

\begin{eqnarray}
\ket{D_{n,k}} = \sum_{\left\{\alpha \right\} } \bigotimes \ket{0}_{i \notin \left\{\alpha \right\}} \bigotimes \ket{1}_{i \in \left\{\alpha \right\} }
\end{eqnarray}
where the summation is over $\left\{\alpha \right\}$, all the distinct $k$ element subset of $\left\{1, \cdots, n \right\}$. 

We show that $\ket{D_{n,k}}$ has tree size $O(n^2)$. To see this, one can consider the uniform superposition of the following Fourier form (omitting normalization):

\begin{eqnarray}
\ket{\psi_{n,k}} = \sum_{j=0}^{k-1} \left( \ket{0}+\exp (\frac{2 \pi i j} {k})\ket{1}  \right)^{\otimes n},
\end{eqnarray}
with tree size $\TS(\ket{\psi_{n.k}}) = O(kn)$.
A direct expansion yields $\sum_{p=0}^n \left( \ket{D_{n,p}}\sum_{j=0}^{k-1} \exp (2\pi i j \frac{p}{k}) \right)$. When $p=mk$, for some integer $m$, $\exp (2\pi i j \frac{p}{k}) = 1$; when $p$ is not a multiple of $k$, $\sum_{j=0}^{k-1}\exp (2\pi i j \frac{p}{k}) = 0$. Hence, $\ket{\psi_{n,k}} = \sum_{m=0}^{\lfloor n/k \rfloor}\ket{D_{n,mk}}$. For $k>n/2$, $m$ can be only 0 or 1, thus $\ket{D_{n,k}} = \ket{\psi_{n,k}} - \ket{0}^{\otimes n}$. For $k=n/2$, $m$ can be 0, 1 and 2; thus $\ket{D_{n,k}} = \ket{\psi_{n,k}} - \ket{0}^{\otimes n} - \ket{1}^{\otimes n}$. For $k<n/2$, by interchanging 0 and 1 we obtain the $k>\frac{n}{2}$ case. So for any $k$, $\TS(\ket{D_{n,k}}) = O(n^2)$. The $n$-qubit W state, which is $\ket{D_{n,1}}$, though representing the most complex class for the three-qubit case, has polynomial tree size $O(n^2)$.

Finally, Matrix Product States (MPS) are a well-studied family because they approximate well the ground state of one-dimensional gapped Hamiltonians~\cite{verstraete2006matrix,perez-garcia2007}. The tree size of an MPS is related to the bond dimension $\chi$. A recursive argument provides the upper bound $O(n^{\log 2\chi})$ for the $\TS$ of an MPS~\cite{us13}: consider the following form of an MPS:
\begin{equation}
\ket{\psi} = \bigotimes_{i=0}^{n}(A_0^{(i)}\ket{0} + A_1^{(i)} \ket{1}),
\end{equation}
where $A_0^{(i)}$ and $A_1^{(i)}$ are matrices of dimension at most $\chi$. By partitioning the qubits into two halves, we have
\begin{equation}
\ket{\psi} = \sum_{s=1}^\chi \ket{\psi^{s,1}_{n/2}} \ket{\psi^{s,2}_{n/2}},
\end{equation}
where now $\ket{\psi^{s}_{n/2}}$ is an MPS of $n/2$ qubits. We can see that $\TS(\ket{\psi_n}) \leq 2\chi \cdot \TS(\ket{\psi_{n/2}})$. By repeating this partitioning, we have $\TS(\ket{\psi_n}) = O((2\chi)^{\log n}) = O(n^{\log 2\chi})$. Thus, if $\chi$ is bounded as $n$ increases, then TS is polynomial. One example of MPS, the 1D cluster state, has $\chi=2$, hence its tree size is $O(n^2)$. Note that the same recursive argument applied to a more general form of tensor network states, the projected entangled pair states (PEPS), gives the superpolynomial upper bound $\chi^{O({\sqrt{n}})}$ and $\chi^{O(n^{2/3})}$ for 2D and 3D PEPS respectively; and indeed, as we are going to discuss later in Sec.~\ref{sec:computation}, PEPS that are universal for measurement-based quantum computation (MBQC) should have superpolynomial tree size, assuming that factoring is not in \textsf{P}.

\section{Complex states}
\subsection{Methods to obtain lower bounds on tree size}
One of the advantages of tree size as a complexity measure is that there are tools for proving lower bound on tree size, hence certifying complex states. One way is to use \textit{counting argument} as Aaronson did in Theorem 7 of~\cite{aaronson2004multilinear}. The fact that there are fewer states that has polynomial tree size than there are in the whole Hilbert space (or the state space of interest), some states are bound to have superpolynomial or even exponential tree size. 

Another method, which will be discussed more often in this paper, is \textit{a theorem first proved by Raz} in the context of multilinear formula size (MFS) \cite{Raz04,aaronson2004multilinear}. Although counting arguments could show that states with superpolynomial tree size must exist, but Raz's theorem allows us to construct explicit examples. Let us present here this important theorem, first in Raz's original formulation, then in an equivalent way in terms of Schmidt rank.

A multilinear formula is a formula that is linear in all of its inputs. The MFS of a multilinear formula is defined as the number of leaves in its \textit{minimal tree representation} similar to the tree size of a quantum state. Consider a multilinear formula  $f:\{0,1\}^n \rightarrow \mathbb{C}$, let $P$ be a bipartition of the input variables $x_1,\cdots, x_n$ into two sets, $y_1,\cdots, y_{n/2}$ and $z_1,\cdots, z_{n/2}$. We now view $f(x)$ as a function $f_P(y,z): \{0,1\}^{n/2} \times \{0,1\}^{n/2} \rightarrow \mathbb{C}$. Then denote by $M_{f|P}$ the $2^{n/2}\times 2^{n/2}$ matrix  whose rows and columns are labeled by $y$ and $z \in \{0,1\}^{n/2}$, respectively. The entry $(y,z)$ of this matrix is defined as $M_{f|P}(y,z) = f_P(y,z)$. Finally, let $\mathrm{rank}(M_{f|P})$ be the rank of $M_{f|P}$ over the complex numbers, and $\mathcal{P}$ be the uniform distribution over all the possible bipartitions $P$. We now state Raz's theorem~\cite{Raz04}:

\begin{theorem}If 
	% \underset
	\label{thm:Raz}
	\begin{eqnarray}
	\underset{P \in \mathcal{P}}{\Pr} \left[ \mathrm{rank}(M_{f|P}) > 2^{\frac{n-n^{1/8}}{2}}\right] = n^{-o(\log n)},
	\end{eqnarray}
	then $\mathrm{MFS}(f) = n^{\Omega(\log n)}$.
\end{theorem}

For any quantum state $\ket{\psi}$, we can define the associated multilinear formula $f_\psi(x) = \braket{x|\psi}$. Note that this formula computes the coefficients in the computational basis expansion of $\ket{\psi}$. Given a tree representation of a quantum state, a tree for the associated multilinear formula can be obtained by interchanging $\ket{0}_i \rightarrow (1-x_i)$ and $\ket{1}_i \rightarrow x_i$. Thus, given the minimal tree of quantum state, we can obtain a tree for the associated multilinear formula with the same size. The true MFS of the formula can only be smaller, therefore~\cite{aaronson2004multilinear,us13}:

\begin{theorem} 
	\label{thm:TSMFS}$\mathrm{MFS}(f_\psi) \leq \TS(\ket{\psi})$. Therefore, if $f_\psi$ satisfies Raz's theorem, then  $\TS(\ket{\psi}) = n^{\Omega(\log n)}$.
\end{theorem}
In fact, in the original paper of Aaronson~\cite{aaronson2004multilinear}, he showed that the other way of the inequality is also true up to $n$, $\TS(\ket{\psi}) = O(\mathrm{MFS}(f_\psi) + n)$. But for the purpose of the rest of this article, $\mathrm{MFS}(f_\psi) \leq \TS(\ket{\psi})$ is sufficient.

For its application in complexity of quantum states, we rephrase Raz's theorem in terms of the Schmidt rank, a well-known concept in quantum information \cite{nielsen2000quantum}:

\begin{theorem} For a pure quantum state of $n$ qubits, consider all the uniformly distributed $(\frac{n}{2},\frac{n}{2})$ bipartitions, if 
	% \underset
	\begin{eqnarray}
	{\mathrm{Pr}} \left[SR > 2^{\frac{n-n^{1/8}}{2}}\right] = n^{-o(\log n)},
	\end{eqnarray}
	where $SR$ is the Schmidt rank of a particular partition, then $\TS = n^{\Omega(\log n)}$.
\end{theorem}

The statement follows indeed from Raz's theorem, because partitioning of the input $x$ of the associated multilinear formula $f_\psi$ is the same as partitioning the qubits of the state $\ket{\psi}$. Note that $M_{f_\psi|P}$ is a matrix each of whose element is a coefficient of the state $\ket{\psi}$ in its computational basis. Thus, the rank of  $M_{f_\psi|P}$ is exactly the Schmidt rank of the state $\ket{\psi}$ for the bipartition $P$~\cite{nielsen2000quantum}.

Interestingly, from the point of view of complex systems and statistical physics, multipartite entanglement were found to be related to the average entanglement across equal bipartitions~\cite{facchi2006,facchi2009,facchi2010classical,facchi2010multipartite}. Instead of the Schmidt rank, the distribution of purity of the partial states over all equal bipartition was studied in those works. This suggests a possible deeper link between tree size complexity and multipartite entanglement.

With the help of these theorems, we shall identify some explicit multiqubit states with superpolynomial tree size.
\subsection{Immanant states}
\label{sec:immanant}
An explicit family of states with superpolynomial tree size can be constructed based on the immanant of a (0,1)~matrix \cite{us13}. Consider the case when the number of qubit is a square number, $n=m^2$, for each bit string $\ket{x}=\ket{x_1,\dots,x_n}$ we arrange the bits $x_1,\dots,x_n$ row by row to an $m\times m$ matrix $M(x)$ such that
\begin{align}
M(x) = 
\begin{pmatrix}
x_{1}     &x_{2} & \cdots & x_{m} \\
x_{m+1} &x_{m+2}    & \cdots & x_{2m} \\
\vdots  &\vdots    & \ddots & \vdots \\
x_{n-m}        &x_{n-m+1}    &  \cdots      & x_{n}
\end{pmatrix},
\end{align}
so $M(x)_{ij}=x_{m(i-1)+j}$. The immanant states are defined in its computational basis expansion as 
\begin{align}
\ket{\mathrm{Imm_n}} = \sum_{x \in \{ 0,1 \}^n} \mathrm{Imm}(M(x)) \ket{x}.
\end{align}
Here the immanant $\rm{Imm}(M)$ of a matrix $M$ is given by
\begin{align}
\mathrm{Imm}(M)=\sum_{\sigma\in S_m}c_{\sigma} \prod_{i=1}^m x_{i\sigma_i},
\end{align}
where $\sigma$ is an element of the symmetric group $S_m$ of all the $m!$ permutation of $\{1,2,\dots,m\}$, and $c_{\sigma}$ is the corresponding complex coefficient. When $c_{\sigma}=1$ for all $\sigma$ the immanant reduces to the permanent, and when $c_{\sigma}=1$ for even permutations and $-1$ for odd permutations it reduces to the determinant. It is proved in Ref.~\cite{us13} that

\begin{theorem}The Immanant states as defined above have $\TS = n^{\Omega(\log n)}$ if the coefficients $c_{\sigma}$ are all nonzero. 
\end{theorem}
The proof relies on Raz's technique~\cite{Raz04} to show that the multilinear formula size of the immanant with nonzero coefficients is superpolynomial, and this theorem follows immediately from Theorem~\ref{thm:TSMFS}.

The permanent and the determinant states,
\begin{eqnarray}
\ket{\mathrm{Perm}_n} &= \sum_{x \in \{ 0,1 \}^n} \mathrm{Perm}(M(x)) \ket{x}, \nonumber \\
\ket{\mathrm{Det}_n} &= \sum_{x \in \{ 0,1 \}^n} \mathrm{Det}(M(x)) \ket{x},
\end{eqnarray}
are two examples in this family of complex states. 

The smallest known formula for computing permanent is the Ryser's formula~\cite{ryser1963}, which is multilinear: let $S$ be one of the $2^m$ subsets of $\{1,2,\dots,m\}$ and $|S|$ the number of its elements, then the permanent of the matrix $M$ is
\begin{eqnarray}
\mathrm{Perm}(M)=\sum_S (-1)^{m+|S|}\prod_{i=1}^{m} \sum_{j\in S} M_{ij}.
\end{eqnarray}
By substituting this formula to the permanent state and carrying out the summation over $x$ we obtain a decomposition with size $n^{\frac{3}{2}} 2^{\sqrt{n}}$. %If we conjectured that Ryser's formula is the smallest multilinear formula that computes the permanent, 
We conjecture that the tree size of the permanent state is $\TS(\ket{\mathrm{Perm}_n})=2^{\Omega(\sqrt{n})}$, see Sec.~\ref{sec:exponential} for detailed discussion.
%based on two strong evidences: computing the permanent of $(0,1)$ matrices is a $\#$\textsf{P}-complete problem \cite{valiant79,arora09};.

A common confusion sometimes arises: why do we treat the permanent state and the determinant state on the same footing while determinant is known to be much easier to compute than the permanent: In fact, there exists a formula that computes the determinant of a $m\times m$ matrix with size $O(m^4)$ \cite{Bird11}. This does not contradict with Raz's result of $m^{\Omega(\log m)}$, since the optimal algorithm does not use a multilinear formula; and only multilinear formulas can be used to find an upper bound on the $\TS$ of the corresponding state.

\subsection{Deutsch-Jozsa states}
The Deutsch-Jozsa algorithm outperforms its classical counterparts in the deterministic case~\cite{nielsen2000quantum}. It is an algorithm that solves the following hypothetical question: A function $f: \left\{ 0,1\right\}^n \rightarrow \left\{0,1\right\}$ is called \emph{balanced} if exactly half of its input is mapped to 0 and the other half to 1, and \emph{constant} if all the inputs are mapped to 0 or 1. Given the promise that the function is either balanced or constant, how many queries do we need to find out whether the function is balanced or constant? Classically, in the deterministic and worse case scenario, it requires $2^{n-1}+1$ queries, in which case the function outputs all 0 or all 1 for the first $2^{n-1}$ queries. 

The Deutsch-Jozsa algorithm solves the quantum version of this problem with only one query, which is exponentially faster than the classical algorithm. In the quantum version, a query is replaced by the quantum oracle $\ket{x}\ket{y} \rightarrow \ket{x}\ket{y\oplus f(x)}$. In this algorithm, one first prepares the input state as $\ket{0}^{n}\ket{1}$, then applies the Hadamard transformation to all the registers, resulting in $2^{-(n+1)/2}\sum_{x \in \left\{0,1\right\}^n} \ket{x} (\ket{0}-\ket{1})$. After applying the oracle, the state becomes $2^{-(n+1)/2}\sum_{x \in \left\{0,1\right\}^n} \ket{x} (\ket{f(x)}-\ket{1\oplus f(x)})$. Since $f(x)$ is either 0 or 1, we can simplify this to $\ket{\psi} = \ket{\psi_{DJ}} \otimes \ket{-}$, where
\begin{eqnarray}
\label{eqn:DJ}
\ket{\psi_{DJ}} = \frac{1}{2^{n/2}}\sum_{x \in \left\{0,1\right\}^n} (-1)^{f(x)} \ket{x}.
\end{eqnarray}
The last qubit register can be left out at this point. Applying the Hadamard transformation to all the qubits once again, we have $2^{-n}\sum_{y \in \left\{0,1\right\}^n} \left( \sum_{x \in \left\{0,1\right\}^n} (-1)^{f(x) + x \cdot y} \ket{y} \right)$, where $x \cdot y$ represents the sum of bitwise product. Finally, a projection onto $\ket{0}^n$ has probability $|2^{-n}\sum_{x \in \left\{0,1\right\}^n} (-1)^{f(x)}|^2$, which evaluates to 1 if $f(x)$ is constant and 0 if $f(x)$ is balanced. This concludes the algorithm, now we switch the focus to the tree size of the state $\ket{\psi_{DJ}}$.

If $f$ is constant, then the state $\ket{\psi_{DJ}}$ is a simple product state $\ket{+}^n$ with tree size $n$. If $f$ is balanced, we would like to show that an overwhelmingly large fraction of balanced functions correspond to states with superpolynomial tree size. Consider a function $f$ randomly drawn from the uniform distribution of all the ${2^n \choose 2^{n/2}}$ balanced functions, let $P$ be a random equal bipartition of the input $x$ into $y$ and $z$, then $M_{f|P}$ the $2^{n/2} \times 2^{n/2}$ matrix whose entries are $(-1)^{f(y,z)}$. Note that for a balanced $f$, the matrix $M_{f|P}$ has exactly half entries equal to $+1$ and half equal to $-1$. Let $\mathcal{E}_1$ be the event that $M_{f|P}$ has full rank $2^{n/2}$, we need to compute the probability that $\mathcal{E}_1$ happens, in order to see whether the balanced function $f$ leads to a state with superpolynomial $\TS$ (c.f. Theorem~\ref{thm:Raz}).

Let us call a matrix with exactly half entries equal to $1$ and the other half $-1$ a \textit{balanced (1,-1) matrix}. Denote by $M_R$ a random balanced (1,-1) matrix, $M_R$ can be chosen by first drawing a random balanced function $f$, then picking a random bipartition $P$ and assigning $M_R=M_{f|P}$. Now let $\mathcal{E}_2$ be the event that $M_R$ has full rank, we have 
\begin{align}
\Pr(\mathcal{E}_2) = \sum_f \Pr(f) \Pr(\mathcal{E}_1|f). 
\end{align}

Next, we split the set of $f$ into those which give rise to a complex state (i.e. satisfy Raz's theorem) and those which do not. Explicitly, let $ C = \left\{f|\Pr(\mathcal{E}_1|f)\geq q \right\}$, where $q$ is a constant to be specified later, and $\bar{C}$ be the complement of $C$, then 
\begin{align}
\Pr(\mathcal{E}_2)=\sum_{C}\Pr(f)\Pr(\mathcal{E}_1|f) + \sum_{\bar{C}} \Pr(f)\Pr(\mathcal{E}_1|f).
\end{align}
Since $\Pr(\mathcal{E}_1|f) \leq 1$ for all $f$ and $\Pr(\mathcal{E}_1|f) \leq q$ for all $f \in \bar{C}$, we have
\begin{align}
\Pr(\mathcal{E}_2) \leq \sum_{C} \Pr(f) + q\sum_{\bar{C}} \Pr(f). 
\end{align}
Note that the sums of the probability that $f$ is chosen from $C$ and $\bar{C}$ give the fraction of states in the respective sets, that is, $\sum_{C}\Pr(f)= \frac{N_C}{N_f} = F_C  $ and $\sum_{\bar{C}}\Pr(f) = \frac{N_f-N_C}{N_f} =1 - F_C $. Substituting this to the above inequality, we arrive at
\begin{equation}
F_C \geq \frac{\Pr(\mathcal{E}_2)-q}{1-q}.
\end{equation}

Thus, to know how large $F_C$ is we need to know $\Pr(\mathcal{E}_2)$, the probability that $M_R$ is invertible where $M_R$ is a random $2^{n/2} \times 2^{n/2}$ balanced (1,-1) matrix. Our numerical evidence shows that $\Pr(\mathcal{E}_2)$ approaches 1 quickly as $n$ becomes large (see Figure~\ref{fig:probE2}). If one believes that $\Pr(\mathcal{E}_2)\approx 1$ for large $n$, which is strongly suggested by the numerical evidence, then by setting $q$ to a constant not close to 1, say 0.5, we see that $F_C \approx 1$. This means that nearly all balanced functions give rise to states with superpolynomial tree size.

One may argue that the large tree size that arises from the Deutsch-Jozsa algorithm has its root in the oracle's access to completely-random balanced function. The link between large tree size and the usefulness of the algorithm is unclear. Nonetheless, this provides us with an example of complex states that appear in a quantum algorithm. More on the relation between state complexity and quantum computation will be discussed in Sec.~\ref{sec:computation}. 
%
%%% Single-column figures in three-column text ranges are defined using the
%%% figure environment:
\begin{figure}[tb]
	%%% Here you can insert almost every LaTeX (PS-Tricks, TikZ, ...) construct you
	%%% need to generate your figure. A very simple way to do it is including an
	%%% external figure using the graphicx package like this:
	\includegraphics[width=0.9\columnwidth]{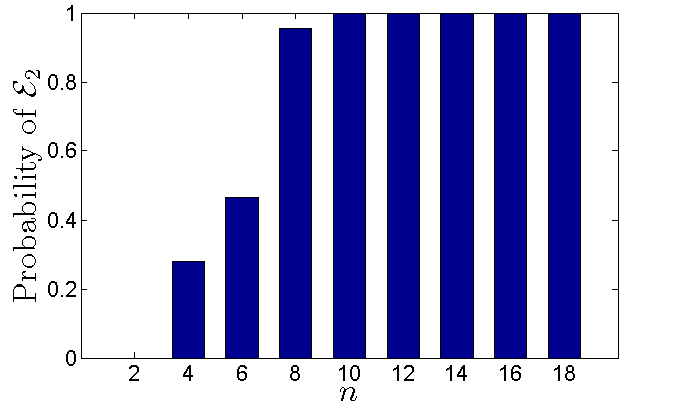}%
	\caption{\label{fig:probE2} %\col 
		%%% \col is to be used here for color figures (only).
		The probability of the event $\mathcal{E}_2$ versus the number of input bits $n$. $\mathcal{E}_2$ is the event that a  random $2^{n/2}\times 2^{n/2}$ balanced (1,-1) matrix has full rank.}
\end{figure}
\subsection{Shor's states}
Shor's algorithm factors an integer $N$ in time $O((\log N)^3)$ \cite{Shor97}, which is exponentially faster than the most efficient known classical algorithm \cite{Pomerance93}. Do states arising from this algorithm have superpolynomial tree size? Aaronson showed the answer is yes assuming a number-theoretic conjecture~\cite{aaronson2004multilinear}. To factorize $N$, pick a pseudo random integer $s < N$, coprime to $N$, consider the Shor's state of $n=\log(N)$ qubits, which is given by
\begin{eqnarray}
\label{eqn:shor}
\frac{1}{2^{n/2}} \sum_{r=0}^{2^n-1} \ket{r} \ket{s^r \mod N}.
\end{eqnarray}
To simplify the proof of the lower bound on the $\TS$ of the Shor's state, it is convenient to measure the second register. Since a measurement in the computational basis does not increase tree size for any outcome of the measurement, we can assume that the measurement outcome to be $1$. Then, the state of the first register has the form
\begin{eqnarray}
\ket{p \mathbb{Z}} = \frac{1}{\sqrt{I}} \sum_{i=0}^I \ket{pi},
\end{eqnarray}
where $p$ is the order of $s$ modulo $N$ and $I=\lfloor(2^n-1)/p \rfloor$. Here $pi$ is represented in binary with $n$ bits, so $\ket{p\mathbb{Z}}$ is a $n$-qubit state. $\TS(\ket{p \mathbb{Z}})$ provides a lower bound for the tree size of the state of the two registers given in~\eqref{eqn:shor}.

The associated formula for this state is a function of a $n$-bit string such that $f_{n,p}(x) = 1$ if $x\equiv 0 \mod p$ and $f_{n,p}(x) = 0$ otherwise. $\mathrm{MFS}(f_{n,p})$ lower bounds $\TS(\ket{p \mathbb{Z}})$, so we shall focus on this formula. Aaronson showed $\mathrm{MFS}(f_{n,p}) = n^{\Omega(\log n)}$ assuming the following number-theoretic conjecture~\cite{aaronson2004multilinear}: there exist constants $\gamma,\delta \in (0,1)$ and a prime $p = \Omega(2^{n^\delta})$ for which the following holds. Let the set $A$ consists of $n^\delta$ elements of $\left\{2^0,\cdots, 2^{n-1}\right\}$ chosen uniformly randomly. Let $S$ consists of all $2^{n^\delta}$ sums of subsets of $A$, and let $S \mod p = \left\{x \mod p: x \in S \right\}$. Then 
\begin{align}
\underset{A}{\Pr} \left[ |S\mod p| \geq (1+\gamma)\frac{p}{2}\right] = n^{-o(\log n)}.
\end{align}

%Based on a number-theoretic conjecture, Aaronson managed to show that for some $\delta \in (0,1)$, there exists $p=\Omega(2^{n^\delta})$ such that $\mathrm{MFS}(f_{n,p})=n^{\Omega(\log_2 n)}$ and hence $\TS(\ket{p\mathbb{Z}}) = n^{\Omega(\log n)}$.
%
%
%
%
\subsection{Subgroup states}\label{complexsubgroup}
%Explicit construction of subgroup state with high Schmidt rank for many partitions.
%
Subgroup states used in quantum error correction also exhibit superpolynomial tree size. Let the element of $\mathbb{Z}^n_2$ be labeled by $n$-bit strings. Given a subgroup $S\subseteq \mathbb{Z}^n_2$, a subgroup state is defined as
\begin{eqnarray}
\ket{S} = \frac{1}{\sqrt{|S|}}\sum_{x \in S} \ket{x}.
\end{eqnarray}

One way to construct a subgroup state is by considering the subgroup to be the null space of a $(0,1)$ matrix over the field $\mathbb{Z}_2$. Given a $n/2 \times n$ binary matrix $A$, a bit string $x$ is in the null space of $A$ if 
\begin{equation}\label{Ax}
Ax=0 \mod 2;
\end{equation} 
and the subgroup state is the equal superposition of all such bit strings. Aaronson shows in Ref.~\cite{aaronson2004multilinear} that, if $A$ is drawn from the set of all possible  $n/2 \times n$ binary matrices, then at least $4\%$ of these matrices give rise to subgroup states with superpolynomial $\TS$.

Let us describe briefly how to prove that a subgroup state has superpolynomial tree size. Consider a random equal bipartition of $x=\{x_1,\dots,x_n\}$ into $y=\{y_1,\dots,y_{n/2}\}$ and $z=\{z_1,\dots,z_{n/2}\}$. Denote by $A_y$ the $n/2 \times n/2$ submatrix of the columns in $A$ that applies to $y$ (see Eq.~\eqref{Ax}), and similarly $A_z$ the submatrix of the columns that applies to $z$. Then, the element of the partial derivative matrix $M_{f|P}(y,z)$ is $1$ when $Ax=0 \mod 2$, which means $A_yy+A_zz=0 \mod 2$, and $0$ otherwise. So long as both $A_y$ and $A_x$ are invertible, for each $y$ there is only one unique value of $z$ that gives $M_{f|P}(y,z) = 1$. In other words, $M_{f|P}$ is a permutation of the identity matrix, hence it has full rank. Based on this observation, one sees that
\begin{theorem}
	\label{thm:epsilon}
	Let $A$ be a $n/2 \times n$ binary matrix and $S = \mathrm{ker}(A)$ over the field $\mathbb{Z}_2$. For random equal bipartitions of $x$ into $y$ and $z$ as described above, if $A_y$ and $A_z$ are both invertible with probability $n^{-o(\log n)}$, then $\TS(\ket{S}) = n^{\Omega(\log n)}$. Moreover, $\TS_\epsilon(\ket{S})=n^{\Omega(\log n)}$ with $\epsilon \leq 1- \mu_n$, where $\mu_n = 2^{-(n/2)^{1/8}/2}$. 
\end{theorem}

\begin{proof}
	The first part follows from the fact that both $A_y$ and $A_z$ being invertible implies that $M_{f|P}$ has full rank. If this happens with probability $n^{-o(\log n)}$, then Raz's theorem is satisfied, hence $\TS(\ket{S}) = n^{\Omega(\log n)}$.
	
	For the second part, we use a lemma proved by Aaronson in Ref.~\cite{aaronson2004multilinear}: Denote by $\ket{\psi}$ a state close to a complex state $\ket{S}$ that satisfies theorem~\ref{thm:epsilon}, such that $|\braket{\psi|S}|^2 \leq 1-\epsilon$. Then, for a fraction of $n^{-o(\log n)}$ of all equal bipartitions, the rank of the partial derivative matrix is
	\begin{eqnarray}
	\mathrm{rank}(M_{\psi|P}) \geq (1-\epsilon)2^{n/2}.
	\end{eqnarray}
	In order to satisfy Raz's theorem, we want $\mathrm{rank}(M_{\psi|P}) \geq 2^{n/2-(n/2)^{1/8}/2}$. A comparison with the above equation gives $\epsilon \leq 1- \mu_n$ where $\mu_n = 2^{-(n/2)^{1/8}/2}$. Therefore, $\TS_\epsilon(\ket{\psi}) = n^{\Omega(\log n)}$ if $\epsilon \leq 1-2^{-(n/2)^{1/8}/2}$. 
\end{proof}
Since $\mu_n$ is exponentially small in $n^{1/8}$, one might think that most states in the Hilbert space satisfy $|\braket{\psi|S}|^2\geq \mu_n$, and hence Theorem \ref{thm:epsilon} can be used to show that most states have superpolynomial $\TS$. This is not correct: Indeed, if $\ket{\psi}$ is randomly and uniformly chosen from the Hilbert space according to a Haar measure, the probability that $|\braket{\psi|S}|^2\geq \mu_n$ is smaller  than $\exp[-(2^n-1)\mu_n]$, which is exponentially small \cite{GFE2009}. However, it is true that most states in the Hilbert space have \emph{exponential} tree size, as showed by a counting argument in Ref.~\cite{aaronson2004multilinear}.

Aaronson first showed an explicit construction by Vandermonde matrix that leads to a superpolynomial complex subgroup state~\cite{aaronson2004multilinear}. Here we present a different construction of the matrix $A$, for which strong numerical evidence suggests that the corresponding subgroup state has superpolynomial $\TS$. Consider the matrix $A_J = \left(\mathds{1}|Q\right)$, where $\mathds{1}$ is the identity matrix and $Q$ a binary Jacobsthal matrix, both of size $n/2\times n/2$. Jacobsthal matrices are used in the Paley construction of Hadamard matrices~\cite{Assmus1992}. The binary version is defined as follows: For a prime number $q$, one can define the quadratic character $\chi(a)$ that indicates whether the finite field element $a \in \mathbb{Z}_q$ is a perfect square. We have $\chi(a) = 1$ if $a = b^2$ for some non-zero element $b \in \mathbb{Z}_q$; and $\chi(a) = 0$ otherwise. Then $Q_{i,j}$ is equal to $\chi(i-j)$. 

We study the partitioning of $A_J$ into $A_y$ and $A_z$ randomly. Numerical evidence (see Fig.~\ref{fig:Jacobsthal}) shows that when \emph{$q$ is a prime and $q=8k+3$ with $k \in \mathbb{N}$}, then , $A_y$ and $A_z$ are both invertible with a probability approaching to a constant around $30\%$. From Theorem~\ref{thm:epsilon} we see that the subgroup state defined by $A_J$ has $\TS=n^{\Omega(\log n)}$ where $n=2q$.

%%% Single-column figures in three-column text ranges are defined using the
%%% figure environment:
\begin{figure}[tb]
	%%% Here you can insert almost every LaTeX (PS-Tricks, TikZ, ...) construct you
	%%% need to generate your figure. A very simple way to do it is including an
	%%% external figure using the graphicx package like this:
	\includegraphics[width=0.9\columnwidth]{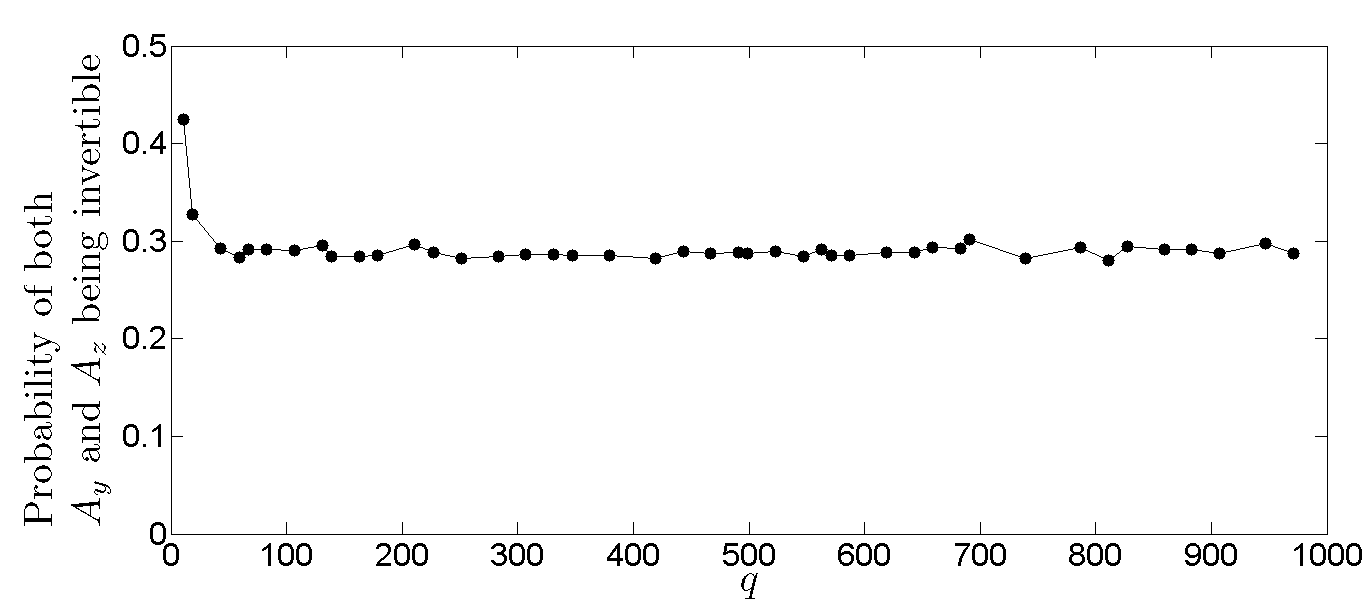}%
	\caption{\label{fig:Jacobsthal} %\col 
		%%% \col is to be used here for color figures (only).
		The probability of both $A_y$ and $A_z$ being invertible over random equal bipartitions of $A_J$. $A_J$ is the $q\times 2q$ matrix $\left(\mathds{1}|Q\right)$, where $Q$ is the Jacobsthal matrix of size $q\times q$, where $q = 3\mod 8$ and is a prime. For large $q$, the probability approaches a constant around $30\%$.}
\end{figure}

\subsection{2D cluster state}
It is known that measurement-based quantum computation (MBQC) on the 2D cluster state is as strong as the circuit model of quantum computation \cite{Briegel012dcluster,Briegel03mbqc}. In this scheme of computation, after the initial resource state is prepared, one only performs single qubit projective measurements and feedforward the outcomes. The power of the computation seems to lie in the initial resource state. Therefore, an initial state that is universal for quantum computation, such as the 2D cluster state, should be highly complex. It is conjectured in Ref.~\cite{aaronson2004multilinear} that the 2D cluster state has superpolynomial TS. By studying the generation a complex subgroup state via MBQC on the 2D cluster state, we can prove that this conjecture is true:
\begin{theorem} 
	The 2D cluster state of $N$ qubits has $\TS=N^{\Omega(\log N)}$. 
\end{theorem}

\begin{proof}
	Suppose we aim to produce an $n$-qubit complex subgroup state $\ket{S_C}$ (as described in Sec.~\ref{complexsubgroup}) that has tree size $n^{\Omega(\log n)}$. These states are known to be stabilizer states~\cite{nielsen2000quantum}. Aaronson and Gottesmannshowed that any $n$-qubit stabilizer state can be prepared using a stabilizer circuit with $O(n^2/\log n)$ number of gates~\cite{Aaronson2004improved} . A stabilizer circuit is one that consists of only {\sc cnot}s, $\pi/2$-phase gates and Hadamard gates. In the MBQC scheme, each of these gates can be implemented by measuring a constant number of qubits: 15 qubits for {\sc cnot}, and 5 qubits for the phase gate and the Hadamard gate~\cite{Briegel03mbqc}. In order to obtain a $n$-qubit complex subgroup state, one needs to prepare a $O(n)$-by-$O(n^2/\log n)$ lattice (see Fig.~\ref{fig:MBQC}), so the number of qubits in the 2D cluster state is $N=O(n^3/\log n)$. Since single qubit projective measurements only decrease tree size (c.f. the proof of theorem~\ref{thm:mbqc}), we have
	\begin{eqnarray}
	\TS(\mathrm{2D \; cluster}) \geq \TS(\ket{S_C}) = n^{\Omega(\log n)} = N^{\Omega(\log N)}.
	\end{eqnarray}
	So, the $N$-qubit 2D cluster state has the superpolynomial tree size.
\end{proof}

\begin{figure}[tb]
	%%% Here you can insert almost every LaTeX (PS-Tricks, TikZ, ...) construct you
	%%% need to generate your figure. A very simple way to do it is including an
	%%% external figure using the graphicx package like this:
	\includegraphics[width=\columnwidth]{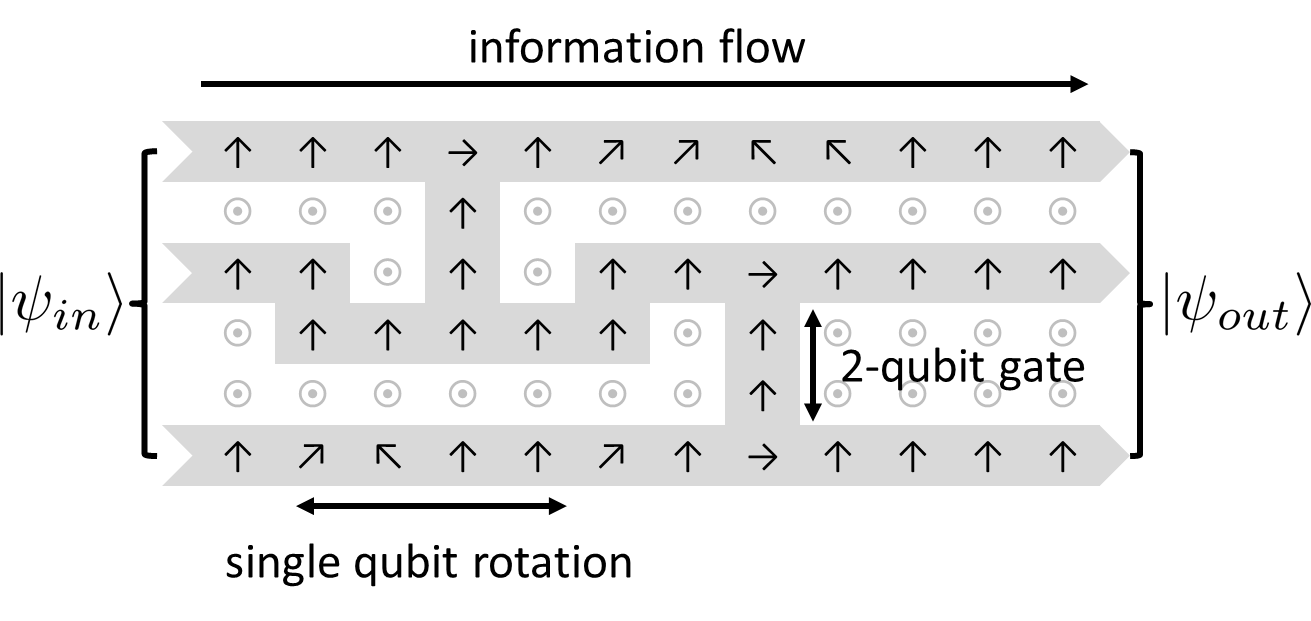}%
	\caption{\label{fig:MBQC} %\col 
		%%% \col is to be used here for color figures (only).
		A schematic diagram of measurement-based quantum computation. Starting from a 2D cluster state, single qubit measurements are performed. $\odot$ represents a $Z$ measurement,  and other arrows represent measurements in the $XY$ plane. The logical input state enters from the left and propagates to the right. Single qubit rotations and controlled gates are realized by a certain sequence of adaptive measurements, from left to right. Implementing a circuit on $n$ qubits with $m$ gates requires a cluster state of size $O(n)$-by-$O(m)$.}
\end{figure}

We can also show that the $\epsilon$-tree size of the 2D cluster state is also superpolynomial:
\begin{theorem}For $\epsilon \leq 1/2$, $\TS_\epsilon(\mathrm{2D\; cluster}) = N^{\Omega(\log N)}$.
\end{theorem}

\begin{proof}
	Assume that we have prepared a state close to the 2D cluster state, $\ket{\mathrm{2D}_\epsilon}$, such that the fidelity $F(\ket{\mathrm{2D}},\ket{\mathrm{2D}_\epsilon}) = |\braket{\mathrm{2D}|\mathrm{2D}_\epsilon}|\geq \sqrt{1-\epsilon}$. Then we apply the same measurement sequence to the erroneous 2D cluster state as if we would to the ideal 2D cluster for preparing a complex subgroup state $\ket{S_C}$. Consider the state after one of the single-qubit measurement in the orthonormal basis $\left\{ \ket{\eta}, \ket{\eta^\perp} \right\}$; the single-qubit projectors are $P_0=\ket{\eta}\bra{\eta}$ and $P_1=\ket{\eta^\perp}\bra{\eta^\perp}$. We now show that one of these outcomes will increase the fidelity between the two cases. If the measurement outcome is not observed, the resulting states on the ideal and $\epsilon$-deviated 2D cluster states are:	
	\begin{align}
	\ket{2D} \rightarrow  \rho&=P_0 \ket{2D}\bra{2D} P^{\dagger}_0+P_1 \ket{2D}\bra{2D} P^{\dagger}_1  \nonumber \\
	&= p_0\ket{\eta}\bra{\eta} \otimes\ket{\psi_0}\bra{\psi_0}  + p_1\ket{\eta^\perp}\bra{\eta^\perp}\otimes\ket{\psi_1}\bra{\psi_1}, \\
	\ket{2D_\epsilon} \rightarrow \sigma &=P_0 \ket{2D_\epsilon}\bra{2D_\epsilon} P^{\dagger}_0+P_1 \ket{2D_\epsilon}\bra{2D_\epsilon} P^{\dagger}_1  \nonumber \\
	&= p'_0\ket{\eta}\bra{\eta}\otimes\ket{\psi'_0}\bra{\psi'_0} + p'_1\ket{\eta^\perp}\bra{\eta^\perp}\otimes\ket{\psi'_1}\bra{\psi'_1},
	\end{align}
	where $\ket{\psi_{0,1}}$ and $\ket{\psi'_{0,1}}$ are the states of the remaining qubits in the cluster; and $p_{0,1}$ and $p'_{0,1}$ are the probability of the measurement outcomes. Clearly, the above map is completely positive and trace-preserving (CPTP). Thus, the fidelity of these two states should not decrease due to monotonicity of fidelity under CPTP maps~\cite{nielsen2000quantum},
	\begin{eqnarray}
	F(\rho,\sigma) \geq F(\ket{\mathrm{2D}},\ket{\mathrm{2D}_\epsilon})= \sqrt{1-\epsilon}.
	\end{eqnarray}
	With a bit of algebra, we can express $F(\rho,\sigma)$ in terms of the fidelity of the post-selected states for the same outcome:
	\begin{eqnarray}
	\label{eqn:fidelity1}
	F(\rho,\sigma) &=& \Tr\sqrt{\rho^{1/2}\sigma\rho^{1/2}} \nonumber\\
	&=& \sqrt{p_0p'_0}|\braket{\psi_0|\psi'_0}|	+ \sqrt{p_1p'_1}|\braket{\psi_1|\psi'_1}|.
	\end{eqnarray}
	Let us denote $x = \max(|\braket{\psi_0|\psi'_0}|,|\braket{\psi_1|\psi'_1}|)$ to be the larger overlap between the two, then
	\begin{eqnarray}
	\label{eqn:fidelity2}
	F(\rho,\sigma) \leq x \left( \sqrt{p_0p'_0} + \sqrt{p_1p'_1} \right) \leq x,
	\end{eqnarray}
	since $\sqrt{p_0p'_0} + \sqrt{p_1p'_1}\leq (p_0+p'_0+p_1+p'_1)/2=1$. Combining Eq.~\eqref{eqn:fidelity1} and Eq.~\eqref{eqn:fidelity2}, we have
	\begin{eqnarray}
	x = \max \left( |\braket{\psi_0|\psi'_0}|,|\braket{\psi_1|\psi'_1}| \right) \geq \sqrt{1-\epsilon}.
	\end{eqnarray}
	Therefore, for at least one of the outcomes, we have a non-decreasing fidelity on the unmeasured parts of the states. For every measurements we post-select on the outcome that do not decrease the fidelity. Note that the complex subgroup states can be realized by a Clifford circuit, which can be implemented by a series of non-adaptive measurements. This means that, regardless of the outcome, the state obtained from the ideal 2D cluster is a complex subgroup state $\ket{S_C}$ upto local Pauli operators. For the erroneous 2D cluster state, we would obtain a state $\ket{S_\epsilon}$ such that $|\braket{S_\epsilon|S_C}| \geq \sqrt{1-\epsilon}$.	From theorem \ref{thm:epsilon}, we see that when $n$ is large enough, $\TS(\ket{S_\epsilon}) = n^{\Omega(\log n)}$, and hence $\TS(\ket{2D_\epsilon})=N^{\Omega(\log N)}$, if $\epsilon \leq 1/2$. Thus, $\TS_\epsilon(\ket{\mathrm{2D}}) = N^{\Omega(\log N)}$ for $\epsilon\leq 1/2$.
\end{proof} 

\section{Witnessing complex states}
In this section we address the problem of verifying the large $\TS$ of complex states. Suppose one wants to create complex states such as the complex subgroup states and the 2D cluster state in the lab, in reality the produced states are at some distance away from the target states due to experimental imperfection. How do we verify that the produced state is superpolynomially complex? Full state tomography requires exponentially many operations and is hence not practical. Nonetheless, \emph{for complex states that are stabilizer states, there exists a complexity witness that can be measured with only a polynomial number of basic operations}. This witness can be used for verifying the superpolynomial $\TS$ of \emph{pure states}. Proving and verifying superpolynomial $\TS$ of \emph{mixed states} remains an open problem.

The subgroups states described in Sec.~\ref{complexsubgroup} belong to the class of stabilizer states. A $n$-qubit stabilizer state $\ket{S}$ is uniquely defined by $n$ mutually commutative stabilizing operators in the Pauli group, $g_1, \dots , g_n$, satisfying the eigenvalue equation:
\begin{eqnarray}
\label{eqn:generator}
g_i \ket{S} = \ket{S}.
\end{eqnarray}
The generators of the subgroups states can be read off from the corresponding matrix $A$. Let $R = \mathrm{rank}(A)$; then there are $R$ linearly independent rows $r_i\;(1\leq i \leq R)$ in $A$. For the first $R$ generators, one simply replace 0 by $I$ and 1 by $Z$ for each of the first $R$ linear independent row. For example, if row $r_i$ is $(0,0,1,0)$, we write $g_i = IIZI$, where the position of the operators denotes the qubit on which they operate on. The remaining generators can be found from the $n-R$ linearly independent vectors $c_i$ that span the null space of $A$. One replaces 0 with $I$ and 1 with $X$ for each vector, and the generator is the ordered product of these operators.

\begin{proposition}The operators $g_i$ defined above are the generators of the stabilizer of $\ket{S}$. 
\end{proposition}
\begin{proof}
	Recall that $\ket{S}$ is the uniform superposition of $\ket{x}$ where $x$ is a vector in the null space of $A$. For the first $R$ generators, we have $g_i\ket{x} = (-1)^{r_i\cdot x}\ket{x} = \ket{x}$, for all $x\in \mathrm{ker}(A)$, hence $g_i \ket{S} = \ket{S}$ for $i=1,\cdots, R$. For the generators obtained from the $n-R$ linearly independent vectors $c_i$ in the null space of $A$, we have $g_i\ket{x} = \ket{x\oplus c_i}$, where $\oplus$ is the bitwise addition modulo 2. Note that $c_i$ is in the null space of $A$, so $\mathrm{ker}(A)+c_i = \mathrm{ker}(A)$, and hence $g_i \ket{S} = \ket{S}$. This shows that the $g_i$s stabilize $\ket{S}$. 
	
	For the commutation relation, it is obvious that the first $R$ generators commutes with each other and so do the $n-R$ obtained from the null space. It remains to show that $g_i$ from row $r_i$ commutes with $g_j$ from $c_j$. $r_i\cdot c_j = 0\mod 2$ implies the number of positions where the entries of both $r_i$ and $c_j$ are 1 must be even. The single-qubit operators in $g_ig_j$ at these positions are $ZX=-XZ$; and since there are an even number of these pairs we see that $g_ig_j = g_jg_i$.
\end{proof}

Now we show how to construct a complexity witness based on the complex subgroup states. Consider a state $\ket{S_C}$ that satisfies Theorem~\ref{thm:epsilon}. For large $n$, any $n$-qubit state $\ket{\psi}$ such that $|\braket{\psi|S_c}|^2\geq 1/2$ must have $\TS=n^{\Omega(\log n)}$. The superpolynomial $\TS$ of these states can be verified by measuring the witness
\begin{eqnarray}
W = \frac{1}{2} \mathds{1} - \ket{S_C}\bra{S_C}.
\end{eqnarray}
A negative value of $\langle W\rangle$ implies that the overlap of the produced state and $\ket{S_C}$ is larger than $1/2$, and hence the $\TS$ of the produced state is superpolynomial. However, $W$ as such is not measurable in practice, under the natural constraint that only local measurements are feasible. If one decomposes $W$ into a sum of locally measurable operators, the number of such measurements increases exponentially with the number of qubits~\cite{Guhne2002,Bourennane04,Sackett00}. Nonetheless, when $\ket{S_C}$ is a stabilizer state, it is possible to construct a \emph{stabilizer witness} $W'$ with the following properties: If $\langle W' \rangle < 0$ then $\langle W \rangle <0$; and $W'$ can be decomposed into a sum of a linear number of operators in the Pauli group, which in turn can be measured by a \emph{polynomial} number of basic operations~\cite{Toth07}. The stabilizer witness is defined as:
\begin{eqnarray}
W' = (n-1)\mathds{1} - \sum_{i=1}^n g_i.
\end{eqnarray}

To show that $\langle W' \rangle<0$ implies $\langle W \rangle <0$, one considers all the eigenvalue equations of the form~\eqref{eqn:generator} but with possible eigenvalues $\pm1$. This defines the set of $2^n$ common eigenstates of the generators $g_i$s. Since all the generators are Hermitian operators, the common eigenstates are mutually orthogonal and form a complete basis. One can verify that, in this basis, the operator $W'-2W$ is a diagonal matrix with non-negative diagonal entries. Thus, $W'-2W$ is a positive semi-definite operator; so $\langle W' \rangle<0$ implies $\langle W \rangle <0$. If in an experiment the expectation value of the stabilizer witness $W'$ is found to be negative, then one can certify that the produced state indeed has $\TS = n^{\Omega(\log n)}$.

While the witness $W$ detects all complex states with a fidelity (with respect to $\ket{S_C}$) larger than $1/2$, $W'$ detects a smaller set. It is necessary to know how close to $\ket{S_C}$ a state $\ket{\psi}$ needs to be for $\bra{\psi}W'\ket{\psi}$ to be negative. If the required fidelity is exponentially close to 1 then no state would be detected by $W'$ in practice. For this purpose, we first expand $\ket{\psi}$ as 
\begin{equation}
\ket{\psi}=c_1 \ket{S_C} + c_2 \ket{S^{\perp}},
\end{equation}
where $\ket{S^{\perp}}$ is a state orthogonal to $\ket{S_C}$ and $|c_1|^2+|c_2|^2=1$. We have 
\begin{equation}
\bra{\psi}W'\ket{\psi}=n-1-n|c_1|^2-|c_2|^2\sum_{i=1}^n \bra{S^{\perp}}g_i\ket{S^{\perp}}.
\end{equation}
Since $1+g_i$ is a positive semi-definite matrix, $\bra{S^{\perp}}g_i\ket{S^{\perp}}\geq -1$. Therefore, 
\begin{equation}
\bra{\psi}W'\ket{\psi}\leq n-1 - n|c_1|^2+n|c_2|^2=2n-1-2n|c_1|^2.
\end{equation}
Thus, $\bra{\psi}W'\ket{\psi}<0$ when the overlap $|\braket{\psi|S_C}|^2=|c_1|^2>1-1/(2n)$. So, the loss of fidelity must be smaller than $1/(2n)$ for a state to be detected by $W'$. 

One needs to measure all the $n$ generators to estimate $\braket{W'}$. With the help of an ancilla qubit, all the generators, each with \emph{two} possible outcomes, can be measured by applying a circuit of size $O(n^2)$ followed by a measurement on the ancilla qubits \cite{nielsen2000quantum} (see Fig.~\ref{measure}). These measurements need to be repeated to obtain the desired accuracy. Suppose the produced state has a fidelity $|\braket{\psi|S_C}|^2=1-\alpha/(2n)$ with $\alpha<1$ is a constant, we have $\braket{W'}< -(1-\alpha)$. If the random error in each $g_i$ is $\delta g$ then $\delta W'=n \delta g$. Thus, to be confident that $\braket{W'}<0$ one needs $n\delta g < 1-\alpha$, or $\delta g<(1-\alpha)/n$, which is achievable with a polynomial $\mathrm{poly}(n)$ number of repetitions. Therefore, a correct negative expectation value of $W'$ can be obtained with polynomial effort. 

\begin{figure}[t]
	\centering
	\includegraphics[scale=0.4]{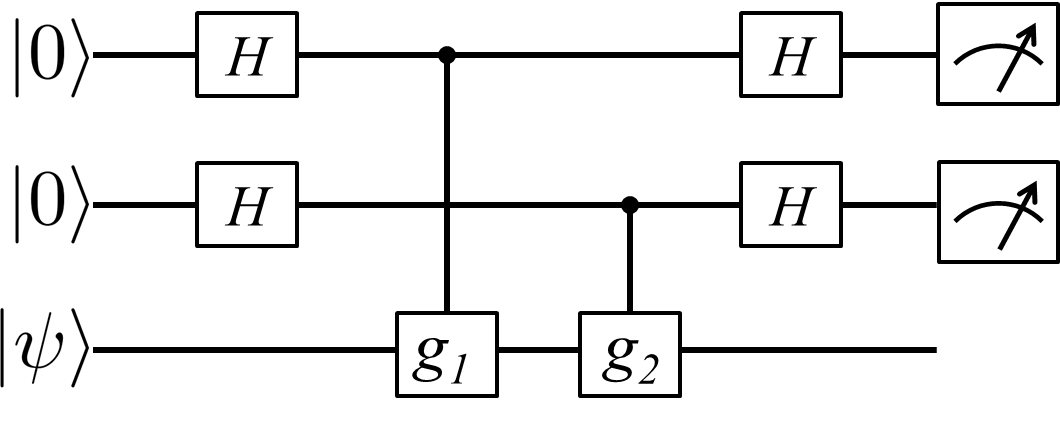}
	\caption{A circuit for measuring all the generators $g_i$'s (only two are shown here). The controlled-$g_i$ gate can be decomposed into at most $n$ two-qubit controlled-Pauli gates, and there are $n$ generators to be measured. Projective measurements of the ancilla qubits in the computational basis give the outcome of $g_i$.}\label{measure}
\end{figure}

There is a similar stabilizer witness for detecting complex states close to the 2D cluster states. Indeed, the 2D cluster state has $\TS_{1/2}=n^{\Omega(\log n)}$ and is also a stabilizer state. Thus, the witness for the 2D cluster state has the same form as $W'$, with the $g_i$s replaced by the generators of the 2D cluster state. These generators are described in Ref.~\cite{Briegel03mbqc}. 

%%%%%%%%%%%%%%%%%%%%%%%%%%%%%%%%%%%%%%%%%%%%%%%%%%%%%%%%%%%%%%%%%%%%%%%%%%%
%%%%%%%%%% Quantum Computation %%%%%%%%%%%%%%%%%%%%%%%%%%%%%%%%%%%%%%%%%%%%
%%%%%%%%%%%%%%%%%%%%%%%%%%%%%%%%%%%%%%%%%%%%%%%%%%%%%%%%%%%%%%%%%%%%%%%%%%%
\section{Relation to quantum computation}
\label{sec:computation}
One of the main motivation of this study is to investigate the relation between state complexity and quantum computation. To elaborate on this, we can divide all the quantum states into four categories according to their preparation complexity and state complexity (see Fig.~\ref{fig:useful}). The set of states with large preparation complexity but small state complexity is presumably empty because preparing simple states should not be too difficult. The states with small state complexity are not useful for quantum computation because they are too simple and hence a classical computer can simulate them efficiently. The states with large preparation complexity are not useful either because quantum computation with these states requires too much resource in space and time. States that are useful for quantum computation should be the ones that have large state complexity yet small preparation complexity. If tree size is a good measure of state complexity, then we might ask: is superpolynomial tree size a necessary condition for the state to provide advantage in some computational task? In this section, we are going to discuss this link in the framework of measurement-based quantum computation and the circuit model of quantum computation.

Note that the complex subgroup states presented in Sec.~\ref{complexsubgroup} belongs to the class of stabilizer states. They have superpolynomial tree size and can be realized by a quantum circuits consist of $O(n^2/\log n)$ number of gates~\cite{Aaronson2004improved}. Therefore, these states belong to the bottom left corner of Fig.~\ref{fig:useful}. But they are not useful for quantum computation since stabilizer circuit can be simulated efficiently on a classical computer~\cite{nielsen2000quantum,Aaronson2004improved,Anders2006fast}.

\begin{figure}[tb]
	%%% Here you can insert almost every LaTeX (PS-Tricks, TikZ, ...) construct you
	%%% need to generate your figure. A very simple way to do it is including an
	%%% external figure using the graphicx package like this:
	\includegraphics[width=0.9\columnwidth]{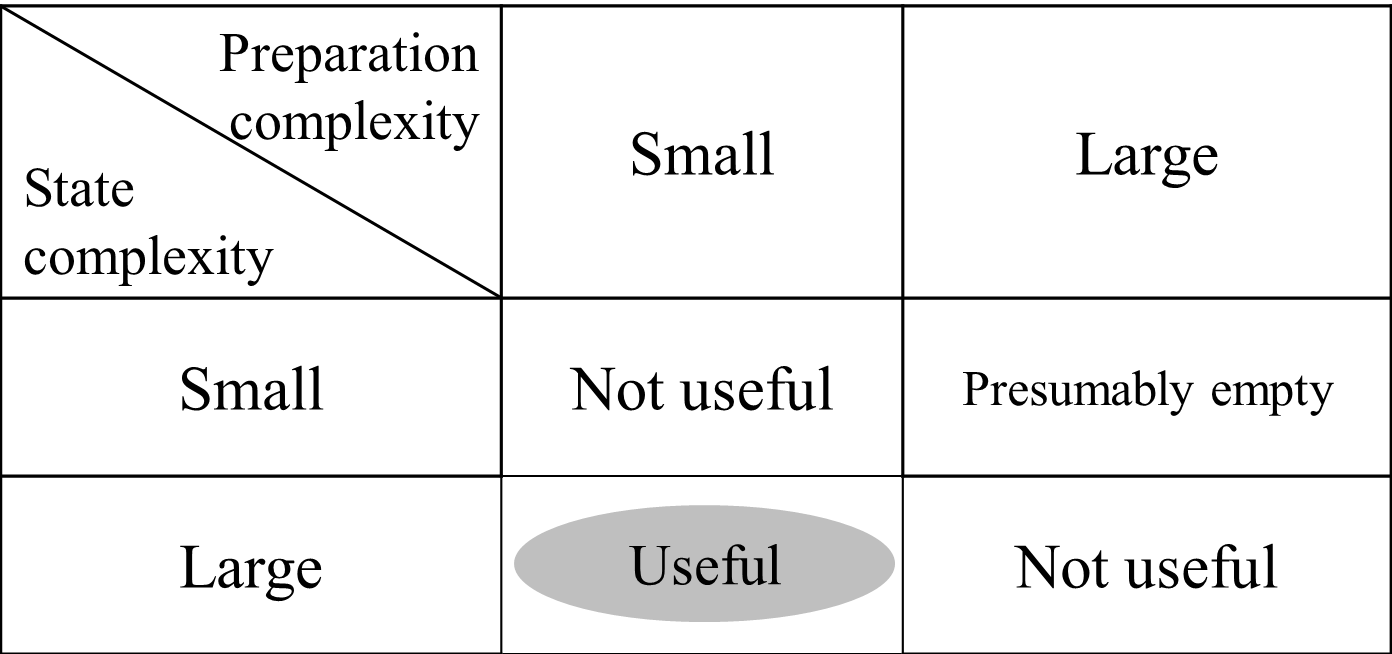}%
	\caption{\label{fig:useful} %\col 
		%%% \col is to be used here for color figures (only).
		Dividing all quantum states into four categories according to their state complexity and preparation complexity. One out of the categories is presumably empty, two are not useful for quantum computation. The states that are useful for quantum computation should have large state complexity and small preparation complexity.}
\end{figure}

\subsection{Measurement-based quantum computation (MBQC)}
\label{sec:MBQC}
There are several theoretical models of quantum computation, including the circuit model and the MBQC model. For the circuit model, the input state can always be the simple product state. The quantum power of the computation lies in the gates applied for coherently manipulating single qubits and entangling different qubits~\cite{nielsen2000quantum}. On the contrary, for MBQC, after the initial resource state is prepared, we perform projective measurements on single qubits and feedforward the results for choosing the basis of the next round of measurement~\cite{Briegel012dcluster,Briegel03mbqc}. Loosely speaking, all the quantum advantage is contained in the resource state. If this resource state is simple, then MBQC will not offer any real speed up over classical computation. To make this intuition more rigorous, we proved that:

\begin{theorem}
	\label{thm:mbqc}If the resource state has $\TS = \mathrm{poly}(n)$, then MBQC can be simulated efficiently with classical computation.
\end{theorem}
\begin{proof}[Proof:]
	Consider the resource state in its minimal tree representation, one sees that at the lowest layer there are a polynomial number of leaves. We will show that it requires polynomial effort to update the tree given a measurement outcome: Assume we measure the $i$th qubit in the basis $\left\{\ket{\eta}, \ket{\eta^\perp} \right\}$ and obtain the the result $\ket{\eta}$, then for every leaf containing qubit $i$, say $c_\alpha \ket{\alpha}_i+ c_\beta \ket{\beta}_i$, we  update it to $(c_\alpha \braket{\eta|\alpha} + c_\beta \braket{\eta|\beta}) \ket{\eta}$. This requires evaluation of the inner products for a polynomial number of leaves. The size of the tree after updating can only get smaller and thus is still polynomial.  So, both the tree representation of the state at each step of the computation and the update of the state after a measurement can be carried out with polynomial effort. It follows that  MBQC on resource states with polynomial $\TS$ can be simulated on a classical computer with polynomial overhead.
\end{proof}

\subsection{Weaker version of the TreeBQP conjecture}
%In the circuit model of quantum computation, if an algorithm provides significant speed up compared with classical computation, we then expect to find a state with superpolynomial tree size in the process of the computation. Now let us look at some examples.

For the circuit model, rather than checking for each algorithm, one would like to have a general proof that small tree size does not provide any computational advantage. In \cite{aaronson2004multilinear}, Aaronson raised the question of whether 
\textsf{TreeBQP~=~BPP}. This remains an open conjecture, here we prove a weaker version of it.

First let us define what \textsf{TreeBQP} is. Bounded-error quantum polynomial-time (\textsf{BQP}) is the class of decision problems solvable with a quantum Turing machine, with at most $1/3$ probability of error. \textsf{TreeBQP} is essentially \textsf{BQP} with the restriction that at each step of the computation, the state is exponentially close to a state with polynomial tree size. In other words, the $\TS_\epsilon$ of the state is polynomial with $\epsilon = 2^{-\Omega(n)}$ (See Eqn.~\eqref{eqn:epsilon}). Since we impose more restrictions, clearly $\mathsf{TreeBQP}\subseteq\mathsf{BQP}$. We can also simulate \textsf{BPP}, the classical counterpart of \textsf{BQP}, in \textsf{TreeBQP}: One simply implements reversible classical computation, applies a Hadamard gate on a single qubit and measures in its computational basis to generate random bits if needed. Since each classical bit string can be represented by a  quantum product state, $\TS$ is $n$ at every steps, so this simulation is in \textsf{TreeBQP}. Thus, we have~\cite{aaronson2004multilinear}:

\begin{theorem}$\mathsf{BPP} \subseteq \mathsf{TreeBQP} \subseteq \mathsf{BQP}$.
\end{theorem}
If $\mathsf{TreeBQP} = \mathsf{BPP}$, then large tree size is a necessary condition for quantum computers to outperform classical ones. Unfortunately, we can only prove a weaker version of this. For this purpose, we first show a proposition that relates tree size and Schmidt rank.

Note that one can draw a rooted tree in a binary form (each gate has only two children) without changing the number of leaves (its size). Next, for any gate $w$ we denote $S(w)$ as the set of qubits in the state described by the subtree with $w$ as the root. Let $Y|Z$ be a bipartition of the qubits into two sets $Y$ and $Z$. A $\otimes$ gate is called \emph{separating with respect to $Y|Z$} when at least one of its children $u$ has the property $S(u) \subseteq Y$ or $S(u) \subseteq Z$. A $\otimes$ gate is called \emph{strictly separating} if its children $u_1, u_2$ satisfy $S(u_1) \subseteq Y$ and $S(u_2) \subseteq Z$. Then,

\begin{proposition}
	\label{thm:separating}For a bipartition of the qubits into $Y$ and $Z$, if there exists a polynomial sized tree such that all the $\otimes$ gates are separating with respect to $Y|Z$, then the Schmidt rank of the state with respect to the bipartition $Y|Z$ is polynomial. 
\end{proposition}
\begin{proof}
	Identify all the strictly separating $\otimes$ gates in the binary tree. Since the number of leaves $N_L$ is polynomial and the total number of gates in the binary tree is $N_G=N_L-1$, the number of strictly separating gates, $N_S$, is also polynomial. It is clearer to look at a representative example in Fig.~\ref{fig:TreeBQP}. Focus on the $+$ gate that joins two such $\otimes$ gates, $\ket{\varphi_{Y_1}}\otimes\ket{\varphi_Z}$ and $\ket{\varphi'_{Y_1}}\otimes\ket{\varphi'_Z}$. Since this $+$ gate contain qubits in both $Y$ and $Z$, and the $\otimes$ gate at the top is separating, the qubits under the sibling of the $+$ gate must be contained strictly in either $Y$ or $Z$. Without lost of generality, let them be contained in $Y$ and denote their state as $\ket{\phi_{Y_2}}$. We can exchange the $+$ gate and the $\otimes$ gate at the top so that the state becomes $(\ket{\varphi_{Y_1}}\otimes\ket{\phi_{Y_2}}) \otimes\ket{\varphi_Z} + (\ket{\varphi'_{Y_1}} \otimes \ket{\phi_{Y_2}}) \otimes \ket{\varphi'_Z}$. Now let us relabel $\ket{\varphi_{Y_1}}\otimes\ket{\phi_{Y_2}}$ as $\ket{\Psi_Y}$ and $\ket{\varphi'_{Y_1}}\otimes\ket{\phi'_{Y_2}}$ as $\ket{\Psi'_Y}$, the state can be written as $\ket{\Psi_{Y_1}}\otimes\ket{\varphi_Z}$ and $\ket{\Psi'_{Y_1}}\otimes\ket{\varphi'_Z}$. The same process can be applied upward until these $+$ gates joins at the root. In the final form of the tree, one sees that the state has a form similar to the Schmidt decomposition: 
	\begin{eqnarray}
	\label{eqn:schmidtlike}
	\ket{\psi} = \sum_{i=1}^{N_\otimes}\ket{\Psi_Y}_i\ket{\Psi_Z}_i,
	\end{eqnarray}
	where $\ket{\Psi_Y}_i$ contain qubits in $Y$ and  $\ket{\Psi_Z}_i$ qubits in $Z$. $N_S$, the number of terms in this Schmidt-like decomposition upper bounds the true Schmidt rank, hence the Schmidt rank is polynomial. 
\end{proof}

Now suppose that at every step of the quantum computation, Proposition\ref{thm:separating} is satisfied for all bipartitions, then the Schmidt rank is polynomial for all bipartitions. It follows from a theorem by Vidal~\cite{vidal2003efficient} that the computation can be efficiently simulated with classical computers.

\begin{figure}[tb]
	
	\includegraphics[width=\columnwidth]{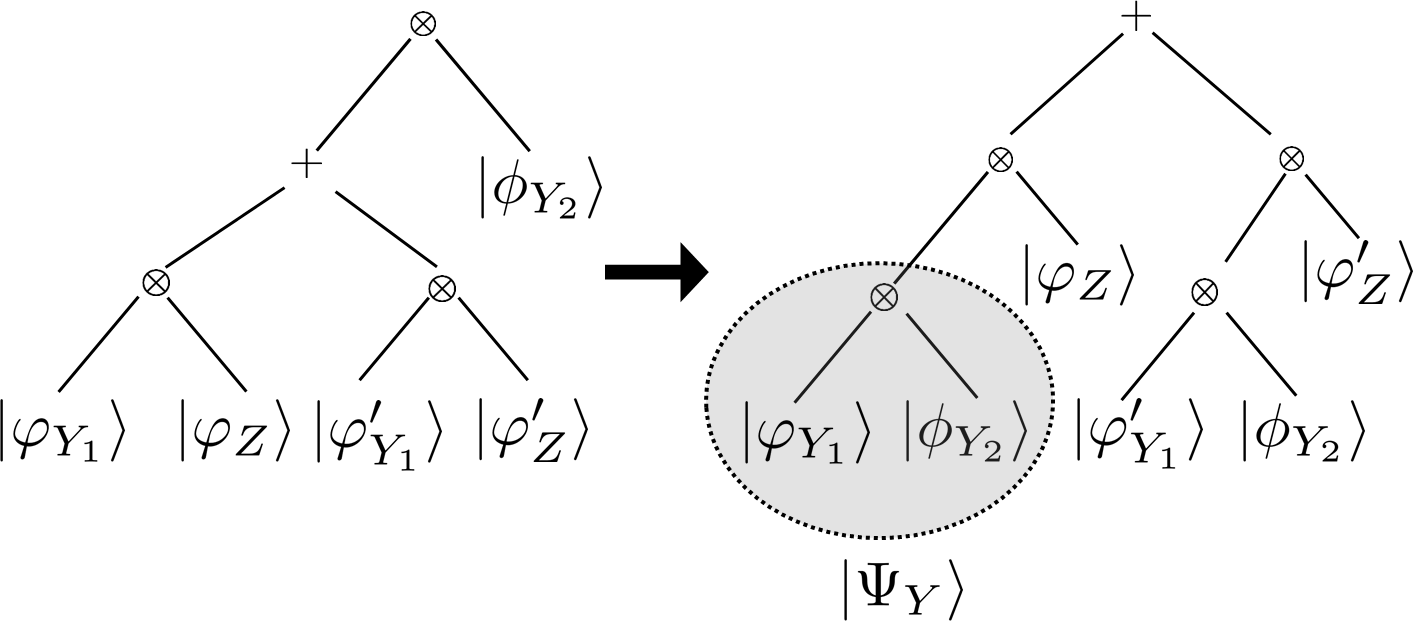}%
	\caption{\label{fig:TreeBQP}%\col 
		%%% \col is to be used here for color figures (only).
		A $+$ gate that joins two $\otimes$ gates with the following property: One of its children are contained in the set of qubit $Y$ and the other contained in $Z$. The sibling $\ket{\phi_{Y_2}}$ of such a $+$ gate must be strictly contained in either $Y$ or $Z$, for $\otimes$ being separating. Now we can exchange the order of $+$ and $\otimes$ by distributing $\ket{\phi_{Y_2}}$ to $\ket{\varphi_{Y_1}}$ and $\ket{\varphi_{Y_1}'}$. This $+$ gate has the same property as before; and this process can be repeated upward until it reaches the root, transforming the tree into a form similar to the Schmidt decomposition.} 
\end{figure}
There are states that do not satisfy the condition of Proposition~\ref{thm:separating}, one example is the optimal tree of the most complex four qubit states (see Eqn.~\eqref{eqn:4qubitdecomp}). There are also states with polynomial $\TS$ that do not satisfy Vidal's criteria, hence do not satisfy Proposition~\ref{thm:separating} for some bipartitions. For example, the state $\left(\frac{ \ket{00} + \ket{11}} {\sqrt{2}}\right) ^ {\otimes n/2}$ has polynomial $\TS$, but there is a bipartition for which the Schmidt rank is $2^{(n/2)}$.
%For an example of a state that does not satisfy the above condition on the $\otimes$ gates, see the optimal tree of the most complex four qubit states (see Eqn.~\eqref{eqn:4qubitdecomp}). The switching of order of qubits in the second branch is the reason why this condition is violated. \textbf{[It is not clear whether a polynomial size tree can always be converted to a form satisfying this condition while remaining polynomial in size.] change to [Any states that does not satisfy Vidal criteria will not satProposition 10. provides a way to relate polynomial tree size with Vidal's criteria, but states that cannot be simulated by Vidal's algorithm will also not satisfy Proposition 10.]}
\section{Open problems}
\label{sec:open}
Even though some properties of tree size have been studied, there are still many open problems remained to be addressed. Here we list a few of the most interesting.

\subsection{On the meaning of the tree size}

Is $\mathsf{TreeBQP}=\mathsf{BPP}$? If this is true, then the role of tree size in quantum computation is clear: Polynomial $\TS$ means efficient classical simulation, and hence large $\TS$ is a necessary condition for quantum speed up.

Is large tree size a resource for any particular task in quantum information? 

Given an explicit family of states, can Raz's theorem be modified to provide an exponential lower bound, instead of $n^{\Omega(\log n)}$?

Is there an algorithm (other than exhaustive search) to find the optimal tree given a quantum state?

How to prove and verify the superpolynomial $\TS$ of mixed states?

Has any (family of) states with superpolynomial tree size been produced in experiments?

As mentioned in Sec.~3, the ground state of a 1D gapped  Hamiltonian, which can be described by an MPS, has polynomial $\TS$. Is there a physically reasonable 1D two-local Hamiltonian whose ground state has superpolynomial $\TS$? A ground state at phase transition no longer obeys the area law because of high entanglement. At this point the state is not an MPS so one can expect that its $\TS$ is large.  
\subsection{Technical open problems}
Below are a few more technical open problems:
\paragraph{Tree size $2^n$ for $n$ qubit states:} One observation we made for the tree size of a few qubits is that the most complex state has $\TS=2^n$ for $n=2,3,4$. Is this a pure coincidence, or is this generally true for any $n$?

\paragraph{Stable tree size:}For the three and four qubit case, the most complex state are unstable. Infinitesimal perturbation in suitable directions in the Hilbert space could reduce its tree size to the second most complex class, hence the  maximal value of the stable tree size $\TS_\epsilon$ is different from the maximal tree size. Is the maximal stable tree size always smaller than the maximal tree size? Is it always equal to the second largest tree size?

\paragraph{Exponential tree size:} \label{sec:exponential}
The $\TS$ of the permanent state and the 2D cluster state are shown to be superpolynomial. We conjecture that they in fact have exponential tree size, $\TS=2^{\Omega(\sqrt{n})}$ and $\TS=2^{\Omega(n^\epsilon)}$ with some $0<\epsilon\leq 1/2$ respectively. There are strong evidences to believe these two states have exponential $\TS$: for the permanent state it is known that computing the permanent of a matrix is $\#$\textsf{P}-hard~\cite{valiant79}, and a subexponential tree size for a permanent state would imply a subexponential formula to compute permanent, contradicting the exponential time hypothesis (or a variant of it, $\#$ETH)~\cite{dell2012exponential}; for the 2D cluster state, if we assume the contrary that it has subexponential tree size, by Theorem~\ref{thm:mbqc} we can simulate the polynomial-time quantum factoring protocol with some subexponential effort. This contradicts with the belief that quantum computing offers exponential speed up compared with classical computing~\footnote{In the case of factoring an integer $N$, Shor's algorithm takes time $O((\log N)^3)$ time while the best classical algorithm takes about $O(e^{1.9(\log N)^{1/3} (\log\,\log N)^{2/3}})$}. Since 2D cluster states are 2D PEPS with a bond dimension $\chi=2$~\cite{verstraete2004}, so $\TS (\ket{2D}) = 2^{O(\sqrt{n})}$ and thus $\epsilon$ is upper bounded by $1/2$.

% based on two strong evidences: computing the permanent of $(0,1)$ matrices is a $\#$\textsf{P}-complete problem \cite{valiant79,arora09}

% the exponential speed up it provides in factoring integers compared with the best known classical algorithm. 

\section{Conclusion}
In this paper, we revisited the complexity measure, called tree size, for pure $n$-qubit states. For few qubits, the state with the largest tree size is identified for $n=2,3,4$ qubits. For 4 qubits, the most compact states admits an optimal tree that is not recursive. The generalization of tree size to incorporate small fluctuation and to mixed state is also discussed.

Raz's theorem on the superpolynomial lower bound of multilinear formula size can be utilized to show that some multiqubit states have superpolynomial tree size. Examples of such complex states, the Immanant states and the subgroup states, are described. Moreover, the conjecture that the 2D cluster state has superpolynomial $\TS$ is proved. We also show how to verify the superpolynomial $\TS$ of stabilizer states, such as the complex subgroup states and the 2D cluster state, with polynomial effort by measuring a stabilizer witness.

The relation between tree size and quantum computation is discussed. In measurement-based quantum computation, if the initial resource state has polynomial tree size, then the computation can be simulated on a classical computer with polynomial overhead. For the circuit model of quantum computation, we show that most of the states arising in Deutsch-Jozsa algorithm have large tree size. A similar result for Shor's algorithm is also reviewed, although a number-theoretic conjecture need to be made in this case. Finally, we present a proof for a weaker version of the $\mathsf{TreeBQP = BPP}$ conjecture, which says that if the tree size of the quantum state is polynomial through out the computation and obeys some extra conditions, then the computation can be simulated efficiently. 

In conclusion, tree size of quantum states as a complexity measure possesses some desirable properties. The most important one is that it is possible to derive non-trivial lower bound on tree size. We have seen some signs on the complex relation between tree size and the usefulness of a state for quantum computation speed up, but the picture is still unclear. By further investigating tree size and other complexity measures, we hope to identify state complexity as the resource for quantum computation. With that understanding one can rule out states which do not provide quantum advantage, and concentrate on producing and characterizing states with high complexity, and possibly identify new quantum algorithms based on complex states.

\section*{Acknowledgments}
We thank an anonymous referee for helpful comments and suggestions on how to prove the superpolynomial tree size of the 2D cluster state, and for showing us that the Dicke states have polynomial tree size. This research is supported by the National Research Foundation Singapore, partly under its Competitive Research Programme (CRP Award No. NRF-CRP12-2013-03) and the Ministry of Education, Singapore. The Centre for Quantum Technologies is a Research Centre of Excellence funded by the Ministry of Education and the National Research Foundation Singapore.

\end{document}